\documentclass[11pt]{amsart}

\usepackage[foot]{amsaddr}

\usepackage{hyperref}
\usepackage{algpseudocode,algorithm,algorithmicx}

\usepackage{amsfonts,amsmath,amssymb,amsthm}
\usepackage{geometry,fullpage}
\usepackage{mdframed}
\usepackage{wrapfig}
\newcommand{\erclogowrapped}[1]{%
\setlength\intextsep{0pt}%
\begin{wrapfigure}[3]{r}{#1*\real{1.1}}%
\includegraphics[width=#1]{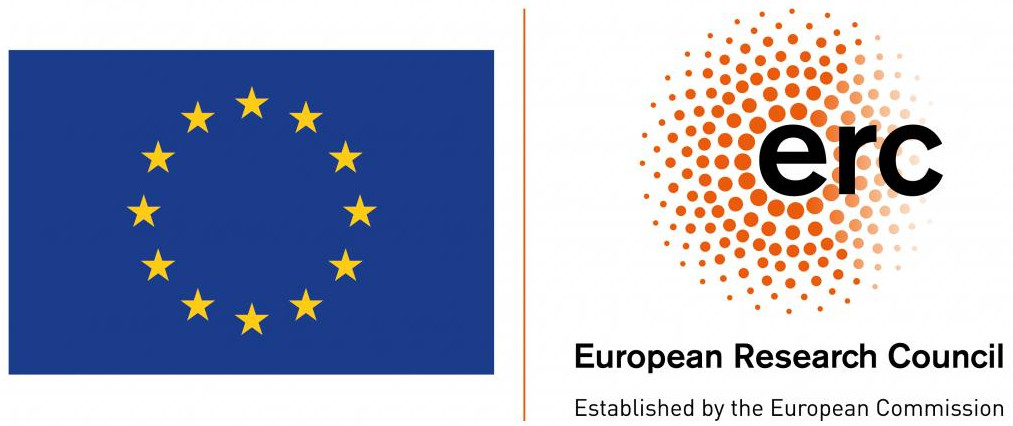}%
\end{wrapfigure}%
}

\usepackage[colorinlistoftodos,prependcaption,textsize=tiny]{todonotes}

\usepackage{amsmath,amsfonts,cleveref}

\newtheorem{lem}{Lemma}[section]

\newtheorem{theorem}[lem]{Theorem}

\newtheorem{fact}[lem]{Fact}
\newtheorem{cor}{Corollary}[section]
\newtheorem{remark}{Remark}[section]

\newtheorem{defn}[lem]{Definition}

\newtheorem{claim}[lem]{Claim}
\newtheorem{definition}{Definition}[section]

\newcommand{\set}[1]{\left\{ {#1} \right\}}
\newcommand{\norm}[1]{{\left\Vert {#1} \right\Vert}}
\newcommand{\paren}[1]{\left( {#1} \right)}
\newcommand{\sparen}[1]{\left[ {#1} \right]}

\newcommand{\op}[1]{\operatorname{#1}}

\newcommand{\rownorm}[1]{\left\| #1 \right\|_{\op{2 \to \infty}}}
\newcommand{\colnorm}[1]{\left\| #1 \right\|_{\op{1 \to 2}}}
 \newcommand{\ip}[2]{\left\langle #1, #2\right \rangle}

\newcommand{\absoluteerror}{\mathsf{err}_{\ell_\infty}}

\newcommand{\ellperror}{\mathsf{err}_{\ell_p}}

\newcommand{\counting}{{\mathsf{count}}}

\newcommand{\ptrace}[1]{\mathsf{Tr}_{\op{p}}(#1 )}
\newcommand{\sliding}{\mathsf{sliding}}

\newcommand{\complex}{\mathbb{C}}
\newcommand{\real}{\mathbb{R}}
\newcommand{\I}{\mathbb{I}}
\newcommand{\E}{\mathop{\mathbf{E}}}
\newcommand{\R}{\mathop{\mathbf{R}}}
\newcommand{\calX}{\mathcal{X}}

\newcommand{\unitary}[1]{\operatorname{U}\left(#1\right)}

\newcommand{\decayingmatrix}{M_f}
\newcommand{\cdsset}{\mathcal{M}}
\newcommand{\dpmechanism}{\mathsf{A}}
\newcommand{\decayingfunctionset}{\mathcal{F}}

\newcommand{\streamlength}{n}

\newcommand{\slidingupperbound}{{1\over 2\streamlength}\sum_{\ell=0}^{2\streamlength-1}\left| {1 - \omega^{W\ell} \over 1-\omega^\ell} \right|}

\newcommand{\countingupperbound}
{1 +  {1 \over \pi} \ln \paren{\streamlength}} 

\newcommand{\multiepochbound}{1 + {1 \over \pi} \ln\paren{
n \over 
b} }
\newcommand{\multi}{\mathsf{striped}}
\newcommand{\mname}{{striped }}


\title{Improved Differentially Private Continual Observation Using Group Algebra}

\author{Monika Henzinger}
\address{Institute of Science and Technology (ISTA), Austria.}
\email{monika.henzinger@ist.ac.at}

\author{Jalaj Upadhyay}
\address{Rutgers University, U.S.A.}
\email{jalaj.upadhyay@rutgers.edu}

\begin{document}

\begin{abstract}
Differentially private weighted prefix sum under continual observation is a crucial component in the production-level deployment of private next-word prediction for Gboard, which, according to Google, has over a billion users. More specifically, Google uses a differentially private mechanism to sum weighted gradients in its \emph{private follow-the-regularized leader} algorithm. Apart from efficiency, the additive error of the private mechanism is crucial as multiplied with the square root of the model's dimension $d$ (with $d$ ranging up to $10$ trillion, for example, Switch Transformers or M6-10T), it determines the accuracy of the learning system. So, any improvement in leading constant matters significantly in practice.

In this paper, we show a novel connection between mechanisms for continual weighted prefix sum and a concept in representation theory known as the \emph{group matrix} introduced in correspondence between Dedekind and Frobenius (Sitzungsber. Preuss. Akad. Wiss. Berlin, 1897) and generalized by Schur (Journal für die reine und angewandte Mathematik, 1904). To the best of our knowledge, this is the first application of group algebra in the analysis of differentially private algorithms. Using this connection, we analyze a class of matrix norms known as {\em factorization norms} that give upper and lower bounds for the additive error under general $\ell_p$-norms  of the matrix mechanism.
This allows us to give 
\begin{enumerate}
	\item the first efficient factorization that matches the best-known non-constructive upper bound on the factorization norm by Mathias (SIAM Journal of Matrix Analysis and Applications, 1993) for the matrix used in Google's deployment, and also improves on the previous best-known constructive bound of Fichtenberger, Henzinger, and Upadhyay (ICML 2023) and Henzinger, Upadhyay, and Upadhyay (SODA 2023); thereby, partially resolving an open question in operator theory,
	\item the first upper bound on the additive error for a large class of weight functions for weighted prefix sum problems, including the sliding window matrix (Bolot, Fawaz, Muthukrishnan, Nikolov, and Taft (ICDT 2013). We also improve the bound on factorizing the \mname matrix used for outputting a synthetic graph that approximates all cuts (Fichtenberger, Henzinger, and Upadhyay (ICML 2023));
    \item a general improved upper bound on the factorization norms that depend on algebraic properties of the weighted sum matrices and that applies to a more general class of weighting functions than the ones considered in Henzinger, Upadhyay, and Upadhyay (SODA 2024). Using the known connection between these factorization norms and the $\ell_p$-error of continual weighted sum, we give an upper bound on the  $\ell_p$-error for the continual weighted sum problem for $p\ge 2$.
\end{enumerate}
\end{abstract}

\maketitle

\thispagestyle{empty}

\clearpage

\renewcommand{\top}{*}
\pagenumbering{arabic} 

\section{Introduction}
To protect the privacy of the input data, current production-level learning systems at Google deploy differentially private algorithms, which, at their core, solve a (high-dimensional extension of a) private \emph{weighted continual (prefix) sums problem}~\cite{mcmahan2022private, dvijotham2024efficient,mcmahan2022federated}. In this problem, a stream of gradients (or numbers) $x_t$  arrives, one at each time step $t$, and the goal is to output the weighted sum of $x_1,\cdots, x_t$. The unweighted version of this problem, called \emph{continual counting}, was introduced by Dwork, Naor, Pitassi, and Rothblum~\cite{Dwork-continual}, and a sequence of papers have given algorithms for it with the goal of improving their accuracy~\cite{andersson2023smooth,andersson2024improved, honaker2015efficient,fichtenberger2023constant,henzinger2023almost}. This problem also has further \emph{practical applications}, such as continual infectious disease incidence counting~\cite{Dwork-continual}, as well as many applications as a \emph{black box algorithm} in other differential private continual algorithms such as histogram estimation~\cite{cardoso2021differentially, chan2012differentially,epasto2023differentially,jain2021price,HenzingerSS24,henzinger2023hist,upadhyay2019sublinear},  graph analysis~\cite{fichtenberger2021differentially, upadhyay2021differentially}, and clustering~\cite{dupre2024}. 

The weighted version of the problem was introduced by Bolot, Fawaz, Muthukrishnan, Nikolov, and Taft~\cite{bolot2013private} and, more recently, Henzinger, Upadhyay, and Upadhyay~\cite{henzinger2024unifying} gave the first general algorithm for this problem. In this problem for any fixed, publicly known \emph{weight function} $f$, the algorithm has to output the weighted sum $\sum_{i=1}^t f(t-i+1) x_i$ in a differentially private manner at each time step $t$. Both works restrict $f$ to be monotonically non-increasing, called the \emph{decayed continual (predicate/prefix) sums problem}. We call the version of the problem without any constraints on $f$ the  \emph{weighted continual sums problem.} In particular, the result in Henzinger, Upadhyay, and Upadhyay~\cite{henzinger2024unifying} has several limitations that impede its usage in practice:
\begin{enumerate}
\label{itemize:issues}
    \item {\bf Monotonicity and positive valued.} The function $f$ has to be a \textit{monotonically non-increasing function}. While some weight functions satisfy this condition, this is not always true for naturally occurring functions. For example, the importance of data changes with time in the financial market, and this change does not exhibit monotonic behavior. %\label{item:monotonic}
    Moreover, the function $f$ has to be \textit{positive real-valued}. While this captures the traditional counting problem, it does not cover many other interesting cases, {such as when $f$ can assume the value 0}  like in the sliding window model and variants. In fact, the original motivation of continual observation, i.e., infectious disease incidence counting, mostly cares about counts in sliding windows~\cite{cdccontact}.
    \label{item:positive} \label{item:monotonic}
    
    \item {\bf Upper bound using Bell's polynomial.} The upper bound on the additive error in \cite{henzinger2024unifying} is stated in terms of Bell's polynomial, which consists of the sum of exponentially many terms and is difficult to compute for many weight functions (even empirically). 
    \label{item:interpretting}
    
    \item {\bf Gap between Upper and Lower Bounds.} 
    The additive error of the continual algorithm depends linearly on the so-called \emph{factorization norms} and \cref{eq:matrixmechanismerror}), which can, thus, be seen as a measure of the ``quality'' of the factorization. 
    For continual counting, the current gap between the upper~\cite{fichtenberger2023constant} and lower bound~\cite{matouvsek2020factorization} on the factorization norm is $\approx 0.68$ for large $n$. Closing this gap is an important problem due to its application in private learning and different areas of mathematics (also see \Cref{rem:operatorTheoryApplication}). 
    \label{item:suboptimality}
\end{enumerate}

In this paper, we address
all these limitations by establishing a close connection between differentially private continual observation and group algebra (Theorem \ref{thm:mainupperboundgamma} and Theorem \ref{thm:mainupperbound}). 

\subsection{Problem Statement.} 
\sloppy 
Let $\mathbb N_+$ denote the set $\set{1, 2, \cdots}$, $\mathbb N$ be the set $\set{0, 1, 2, \cdots}$, and $\real$ denotes the set of real numbers.  
Let $\decayingfunctionset :=\set{f : \mathbb N \to \mathbb R}$ be the set of non-negative valued functions. 
%As observed in previous work~\cite{henzinger2024unifying}, c
Continually observing all $\streamlength$ values (exactly) corresponds to  computing $\decayingmatrix x$, where $x$ is the $\streamlength$-dimensional vector formed by the stream $x_1, x_2, \dots, x_\streamlength$ and $\decayingmatrix$ for $f \in \decayingfunctionset$ is defined as the following lower-triangular {\em Toeplitz} matrix, i.e, a matrix whose descending diagonals from left to right are constant:
\begin{align}
\decayingmatrix = \begin{cases}
    f(i-j) & i \geq j \\
    0 & \text{otherwise}
\end{cases}.
\label{def:Mf}
\end{align}

We use $\cdsset$ to denote the set of Toeplitz matrices formed as above by functions $f \in \decayingfunctionset$.

Our approach is to use the {\em matrix} (or {\em factorization}) {\em mechanism} with $\decayingmatrix$~\cite{edmonds2020power, li2015matrix}, but we give a different factorization from Henzinger, Upadhyay, and Upadhyay [34], resulting in better accuracy (i.e., smaller additive error), a more general weight function class $\decayingfunctionset$, and simpler upper bounds on several error metrics. 
It has been shown that, depending on the error metric considered, the error of the matrix mechanism can be measured by different factorization norms~\cite{edmonds2020power,liu2024optimality,nikolov2023gaussian}. To define the factorization norms, we first need to define the \textit{generalized $p$-trace}~\cite{nikolov2023gaussian}. For a complex-value matrix $M \in \mathbb C^{d \times n}$, let $N=MM^*$, where $M^*$ denotes the complex conjugate of the matrix $M$. Then, for $p\in[2,\infty)$, the generalized $p$-trace of a matrix $M$ is defined as 
\begin{align*}
    \ptrace{M} =  \left(\sum_{i=1}^d \paren{N[i,i]}^{p/2} \right)^{1/p}. %, \quad \text{where} \quad C=AA^\top 
\end{align*}

For any $a,b \in \mathbb N$, let $ \norm{M}_{a \to b} = \max_x \frac{\norm{Mx}_b}{\norm{x}_a}$. Then by standard limiting argument in analysis, we have $$\mathsf{Tr}_\infty (M)  = \max_{1 \leq i \leq d} N[i,i]^{1/2} = \rownorm{M}$$ by the definition of row norm of matrix $M$.
We can now define the class of {\em factorization norms} for all $p\in[2,\infty)$:
\begin{align}
\label{eq:generalfactorizationnorm}    
\gamma_{(p)}(M) = \inf \left\{ \ptrace{B} \cdot \colnorm{C}: M = BC  \right\}. 
\end{align}

The minimization problem in \cref{eq:generalfactorizationnorm} has a natural convex optimization perspective~\cite{nikolov2023gaussian}, 
where the output are the matrices $B$ and $C$ that certify the value of $\gamma_{\op{(p)}}(M)$.
However, currently known optimizers do not scale to $n=1$ million, a setting common in modern data analysis or machine learning\footnote{Denisov et al.~\cite{mcmahan2022private} 
published the first work that proposed using the matrix mechanism in private machine learning. It
computes a factorization that minimizing $\gamma_{\op{(2)}}(\decayingmatrix)$ when $f(i)=1$. Even that work only reports results for $n\leq 4096$ (see their Table 2). We believe it is mainly due to the scalability issue of the optimization algorithm.}.
As the $\ell_p$-error of the matrix mechanism (see \cref{eq:ellError}) is upper bounded by the $\gamma_{\op{(p)}}$ norm~\cite{liu2024optimality,nikolov2023gaussian}, by considering all $p \in [2,\infty)$, one can gain an understanding of the additive error not only in expectation but also with respect to the probability tail bounds. Therefore, a natural goal is to find matrices $B$ and $C$ that achieve an optimum $\gamma_{\op{(p)}}$ value without the need to solve a convex program. 

The two most commonly used factorization norms are for $p=\set{2,\infty}$ (even though the notation is counterintuitive):
$\gamma_{\op{F}}(M) := \gamma_{(2)}(M)$ and $\gamma_{\op{2}}(M) := \gamma_{(\infty)}(M).$
We note that the error of private machine learning algorithms based on matrix mechanism is stated in terms of these two norms~\cite{mcmahan2022private, dvijotham2024efficient}.  To avoid the need for deploying an optimization algorithm, a significant research effort in the recent past has been to find and analyze explicit factorizations with as tight as possible bounds on factorization norms~\cite{dvijotham2024efficient,fichtenberger2023constant,henzinger2024unifying,liu2024optimality}. 

\subsection{Main Result.}
Our main result is a new explicit factorization for $\decayingmatrix$ {for all $\decayingmatrix \in \cdsset$} that addresses the issues mentioned on page \pageref{itemize:issues}: (i) Our bound for the additive error of our factorization is for \emph{all} real-valued weight functions $f$ and not only for monotonic positive-valued functions, thereby addressing concerns mentioned in  \cref{item:positive}; (ii) for \cref{item:interpretting}, our bound is in terms of the algebraic property of the matrix, making it (arguably) \emph{more interpretable}; and (iii) our bounds improve upon the bounds in earlier work and matches the best-known non-constructive bound for one of the most important matrices, namely the unweighted one, partially resolving \cref{item:suboptimality}.

More specifically, our main bound on the factorization norm is stated in terms of the evaluation of a specific polynomial defined uniquely by the entries specifying the Toeplitz matrix.  In our case, for matrix $\decayingmatrix \in \cdsset$, we define a polynomial whose coefficients are entries of $M_f$: 
\begin{align}
m_f(x) = f(0) + f(1) x + f(2)x^2 + \cdots + f\paren{n-1}x^{n-1}.
\label{eq:decaypolynomialHenzingerUU24}
\end{align}

\begin{theorem} 
\label{thm:mainupperboundgamma}
Given a function $f:\mathbb N \to \mathbb R$ and an $n \in \mathbb N_+$, for the matrix $\decayingmatrix$ as defined above, there is an efficient algorithm, described as \Cref{alg:factorization}, that outputs a factorization of $\decayingmatrix = LR$ such that $L$ is a lower-triangular matrix. Let $m_f(\cdot)$ be the polynomial defined in \cref{eq:decaypolynomialHenzingerUU24} and $\omega=e^{\pi \iota/\streamlength} \in \complex$ be the $2\streamlength$-th root of unity, i.e., the solutions of $z^{2\streamlength}-1=0$. Then for all integers $p\in [2,\infty)$,   
 \begin{align}
 \begin{split}
     \gamma_{\op{(p)}}(\decayingmatrix) \leq \ptrace{L}\colnorm{R} \leq  {1 \over 2\streamlength^{1-1/p}} \sum_{k=0}^{2\streamlength-1}\left| m_f(\omega^k) \right|. 
\end{split}
 \label{eq:mainupperboundgammanorm}
 \end{align}  
\noindent In particular, $\gamma_{\op{2}}(\decayingmatrix) = \gamma_{(\infty)}(\decayingmatrix) \leq {1 \over 2\streamlength} \sum_{k=0}^{2\streamlength-1}\left| m_f(\omega^k) \right|$ and $\gamma_{\op F}(\decayingmatrix)  \leq {1 \over 2\sqrt{\streamlength}} \sum_{k=0}^{2\streamlength-1}\left| m_f(\omega^k) \right|$.
\end{theorem}

\begin{remark}
[Generality of approach]
     For the ease of presentation and calculation for special cases mentioned below, in Theorem~\ref{thm:mainupperboundgamma} and throughout this paper, we consider the multiplicative group formed by the $2\streamlength$-th roots of unity under the binary operation of complex multiplication. However, our result can be easily extended to any multiplicative group $G$ of order $2\streamlength$ with $\omega$ replaced by any generator, $g$ of a group $G$ or order $2\streamlength$.  
\end{remark}

\begin{remark}
    [Known vs Unknown $n$] 
    \label{rem:unknown}
    \Cref{alg:factorization} relies on knowing the value of $n$ ahead. In contrast, Henzinger, Upadhyay, and Upadhyay~\cite{henzinger2024unifying} do not require the value of $n$. One can use the technique proposed in Chan et al.~\cite{chan2011private} in conjunction with \Cref{alg:factorization} to deal with unknown $n$, but that comes at the cost of increasing the constant. 
\end{remark}

To compare our  bounds, we consider the cases for which closed-form bounds are known:

\textbf{Case 1: $f(i)=1$}. In this case, $M_\counting := \decayingmatrix$, and, we have 
$\gamma_{\op 2}(M_\counting) \leq \countingupperbound$ (\Cref{sec:specialMatrices}).  
Before our work, the best-known constructive upper bound (i.e., an upper bound for which a factorization is known)
was $\gamma_2(M_\counting) \leq 1 + {\gamma + \ln(n) \over \pi}
$~\cite{fichtenberger2023constant}, where $\gamma \approx 0.58$ is Euler-Mascheroni's constant. Note that our bound is an additive term $\gamma/\pi$ better than that constructive bound, and matches the best known non-constructive bound~\cite{mathias1993hadamard}. Recently, Dvijotham, McMahan, Pillutla, Steinke, and Thakurta~\cite{dvijotham2024efficient} gave a factorization that is a  (multiplicative) $1+o(1)$ factor worse than that of Fichtenberger, Henzinger, and Upadhyay~\cite{fichtenberger2023constant}, but is more space and time efficient. On the other hand, the best-known lower bound  is  
$\gamma_{\op 2}({M_\counting}) \geq {\ln((2\streamlength+1)/3) + 2 \over \pi}$~\cite{matouvsek2020factorization} (also see \cite{dvijotham2024efficient} for a discussion).   
Note that there is still a small additive gap between our upper bound on $\gamma_{\op 2}(M_\counting)$ and the best known lower bound. 
Closing this gap is an interesting open problem.

\begin{remark}
\label{rem:operatorTheoryApplication}
    Computing the exact values of $\gamma_{\op{(p)}}(M_\counting)$ has long intrigued operator theorists due to its application in non-commutative matrix algebra~\cite{junge2005best}, symplectic capacity~\cite{gluskin2019symplectic}, compact operators~\cite{aleksandrov2023triangular,kato1973continuity}, absolute summing problems~\cite{gordon1974absolutely}, etc. This has led to a significant effort in computing the exact value of its factorization norms~\cite{bennett1977schur,davidson1988, kwapien1970main, mathias1993hadamard}. 
    Theorem~\ref{thm:mainupperboundgamma} makes progress on this question as it reduces the gap between the upper and lower bound for $\gamma_{\op{2}}(M_\counting)$. 
\end{remark}

\sloppy
{\textbf{Case 2: Sliding window and striped settings.} 
Next, we consider the weight function of  (a) the sliding window model with window size $W$, and (b) of the $b$-\mname matrix for a positive integer $b$ that is used by Fichtenberger, Henzinger, and Upadhyay~\cite{fichtenberger2023constant} to output,  in the continual setting, a synthetic graph that approximates all graph cuts up to an additive error.
They are defined, respectively, as follows: 
 \begin{align}
 M_\sliding[i,j] = \begin{cases} 1 & 0 \leq i-j \leq W \\
 0 & \text{otherwise} 
 \end{cases}, \quad \text{and} \quad M_\multi[i,j] = \begin{cases} 1 & i\equiv j \mod b \text{ and } j \leq i \\
 0 & \text{otherwise} 
 \end{cases}.
 \label{eq:sliding}
 \end{align}

Our factorization achieves $\gamma_{\op 2}(M_\sliding) \leq  \slidingupperbound$, where $\omega$ is the $2\streamlength$-th root of unity, and $\gamma_{\op{2}}(M_\multi) \leq \multiepochbound$ (see 
Section~\ref{sec:specialMatrices} for a proof).
For $M_\sliding$, none of the previous results compute the $\gamma_{\op 2}$ norm either constructively or non-constructively. Thus, our theorem is the first result that gives an explicit factorization and upper bounds for $M_\sliding$.  We also show in Theorem~\ref{thm:lowerSlidingWindow} that $\gamma_{\op 2}(M_\sliding) \geq {\ln((2W+1)/3) + 2 \over \pi}$. For $M_\multi$, the best prior upper bound was $1 + {\gamma + \ln(n/b) \over \pi}$~\cite{fichtenberger2023constant}, where $\gamma \approx 0.58$ is Euler-Mascheroni's constant. Similarly, using Matousek et al.~\cite{matouvsek2020factorization} and Haagerup~\cite{haagerup1980decomposition}, we can observe that  
$\gamma_{\op 2}({M_\multi}) \geq {\ln((2\streamlength/b+1)/3) + 2 \over \pi}$. 

\subsection{Differentially Private Continual Observation.}
In the static setting, algorithms for computing the product of a public matrix $M$ and a privately-given vector $x$ are well studied, and their quality is usually measured by the {\em (additive) mean-squared error} (aka $\ell_2^2$-error) and the \emph{(additive) absolute error} (aka $\ell_\infty$-error).  There has been some recent interest in studying $\ell_p$-error metric as well~\cite{aumuller2023plan,lebeda2023diffrentially,liu2024optimality,nikolov2023gaussian} for any $p\in [2,\infty)$. In this paper, we study the $\ell_p$-error for any $p\in [2,\infty)$. The (additive) $\ell_p$-error of a randomized algorithm $\mathsf{A}(M;\streamlength)$  for computing $Mx$ on any real input vector $x\in \real^\streamlength$ is defined as 
\begin{align}
\label{eq:ellError}
\ellperror(\mathsf A(M;\streamlength))
= \max_{ x \in \mathbb R^\streamlength} \paren{\E_{\mathsf A} \sparen{\norm{\mathsf{A}(M;\streamlength)(x) - M x}_{p}^p}}^{1/p}.
\end{align}

We use {\em differential privacy} as the notion of privacy: 
\begin{definition}
[Differential privacy~\cite{dwork2016calibrating}]
\label{defn:dp}
Let $\mathsf A : X \rightarrow R$  be a randomized algorithm mapping from a domain $X$ to a range $R$. $\mathsf{A}$ is $(\epsilon,\delta)$-differentially private if for every all neighboring dataset $x$ and $x'$ and every measurable set $S \subseteq R$,
$$\mathsf{Pr} [\mathsf{A} (x) \in S] \leq e^\epsilon \mathsf{Pr} [\mathsf{A} (x') \in  S] + \delta.$$
\end{definition}

The above definition relies on the notion of neighboring datasets. We use the standard notion of neighboring datasets.  For continual observation,  two streams, $S = (x_1,\cdots, x_\streamlength) \in \real^\streamlength$ and $S' = (x_1',\cdots, x_\streamlength') \in \real^\streamlength$ are neighboring if there is at most one $1 \leq i \leq n$ such that $x_i \neq x_i'$. This is known as {\em event level privacy}~\cite{chan2011private,Dwork-continual}.

Our work extends and improves the recent works that used the matrix mechanism for performing continual observation~\cite{choquette2022multi, mcmahan2022private, dvijotham2024efficient,fichtenberger2023constant,henzinger2023almost, henzinger2024unifying}. Given a linear query matrix $M \in \real^{q \times n}$, the {\em matrix} (or {\em factorization}) mechanism~\cite{edmonds2020power,li2015matrix} first computes factors $L$ and $R$ such that $LR=M$. The answer to the linear query on a given input $x \in \real^\streamlength$ is then computed by returning $\mathsf A_{\op{fact}}(M;\streamlength) (x)= L(Rx+z) = Mx + Lz$ for an appropriately scaled Gaussian vector $z \sim \mathcal N(0, \sigma_{\epsilon,\delta}^2 \colnorm{R}^2 \mathbb I)$, where $\sigma_{\epsilon,\delta} = \frac{2}{\epsilon} \sqrt{{\frac{4}{9} +  \ln \paren{\frac{1}{\delta}\sqrt{\frac{2}{\pi}}}}}$ to preserve $(\epsilon,\delta)$-differential privacy~\cite{dwork2006our}. Using the concentration of Gaussian vectors~\cite{liu2024optimality}, we have, for any $q \in \mathbb N_{+}$ and $p\in [2,\infty)$,
\begin{align}
    \ellperror(\mathsf A(M;q)) \leq \sigma_{\epsilon,\delta}\gamma_{(p)} (M) \min\{p,{\ln^{1/2}(q)} \}. 
\label{eq:matrixmechanismerror}
\end{align}

Together with Theorem \ref{thm:mainupperboundgamma}, \cref{eq:matrixmechanismerror} directly implies the following result:

\begin{theorem}
\label{thm:mainupperbound}
For any $0 <\epsilon,\delta <1$ and $M_f \in \mathcal M$, there is an efficient  $(\epsilon,\delta)$-differentially private continual counting algorithm $\mathsf A_{\op{fact}}$, described in \Cref{alg:factorizationmechanism}, that, on receiving a stream of values $x_t$ from $[-\Delta, \Delta]$ of length $\streamlength$ and a weight function $f:\mathbb N \to \mathbb R$, achieves the following error bound for every  $p\in [2,\infty)$:  
 \begin{align}
 \begin{split}
    \ellperror(\mathsf A_{\op{fact}}(\decayingmatrix;\streamlength)) \leq {\sigma_{\epsilon,\delta}  \min\{\sqrt p,\sqrt{\ln(n)}\}\over 2\streamlength^{1-1/p}} \sum_{k=0}^{2\streamlength-1}\left| m_f(\omega^k) \right|, \quad \text{where}~\sigma_{\epsilon,\delta} = {2 \Delta  \over \epsilon} \sqrt{\ln\paren{1.25\over\delta}}
\end{split}
 \label{eq:mainupperbound}
 \end{align}  
 is the variance required by the Gaussian mechanism to preserve $(\epsilon,\delta)$-differential privacy. 
\end{theorem}

\begin{cor}
\label{cor:continualcounting}
Running $\mathsf A_{\op{fact}}$ on a stream of values from $[-\Delta, \Delta]$ of length $\streamlength$, we have (i) for 
continual counting $\absoluteerror(\mathsf A_{\op{fact}}(M_\counting;\streamlength)) \leq {\sigma_{\epsilon,\delta}} \paren{ \countingupperbound } \sqrt{\ln(n)}$; (ii) for the sliding window model 
    $\absoluteerror(\mathsf A_{\op{fact}}(M_\sliding;\streamlength)) \leq {\sigma_{\epsilon,\delta}} \paren{ \slidingupperbound } \sqrt{\ln(n)}$; and (iii) for the striped setting 
    $\absoluteerror(\mathsf A_{\op{fact}}(M_\multi;\streamlength)) \leq {\sigma_{\epsilon,\delta}} \paren{  \multiepochbound} \sqrt{\ln(n)}.$
 \label{eq:maincontinual}
 \label{eq:mainsliding}
\end{cor}
Our algorithm achieves better accuracy for continual counting than earlier explicit factorizations~\cite{dvijotham2024efficient,fichtenberger2023constant,henzinger2024unifying} owing to the improvement on $\gamma_{\op 2}(M_\counting)$. For the sliding window, the algorithms in~\cite{bolot2013private,henzinger2024unifying} used binary counting as a subroutine to handle the sliding window setting to get an additive error of ${4\Delta  \over \epsilon} \paren{1 + {\gamma+\ln(2W) \over \pi} }\sqrt{\ln\paren{2.5\over\delta}\ln(n)}$. As we directly factorize $M_\sliding$, we gain {more than} a multiplicative factor of $2$ over the algorithm in \cite{bolot2013private, henzinger2024unifying}. 
For $M_\multi$ our new factorization improves over prior work~\cite{fichtenberger2023constant} by the same amount as 
we improve the 
$\gamma_{(p)}$ norm. 

\smallskip
\paragraph{\textbf{Space consideration.}} 
Dvijotham et al.~\cite{dvijotham2024efficient} recently proposed an algorithm that outputs a factorization in $O(\log^2(n))$ space for $M_\counting$ while achieving {only slightly worse accuracy than us}. So a natural question is whether one can achieve a low-space algorithm for weighted continual sum problems, too, as considered in this paper. We next show that it is not the case if we consider pure additive error (ie no multiplicative error) with an error that scales inversely with $\epsilon$.

Suppose for contradiction that for some $\epsilon>0$ and $\delta\in (0,1)$ there exists an $(\epsilon,\delta)$-DP algorithm $\mathsf A$ for a general weighted function $\decayingfunctionset$, as considered in our paper, satisfying the following guarantees: $\mathsf A$ uses $o(n)$ space and incurs a non-trivial additive error that is inversely proportional in $\epsilon$ and sublinear in $n$. Then, by letting $\epsilon \to \infty$, we get a non-private exact algorithm that uses $o(n)$ space. This contradicts the lower bounds for any non-private algorithms for exponentially weighted matrices (Cohen and Strauss \cite[Lemma 3.1]{cohen2003maintaining}) and polynomial weighted matrices (Cohen and Strauss \cite[Lemma 3.2]{cohen2003maintaining}).

\section{Preliminaries}
\label{sec:preliminaries}
This section contains the necessary basic definitions to the level of exposition required to understand this paper.

\subsection{Field and Galois Theory.}
A {\em group} $G=(S,\odot)$ is a set $S$, together with a binary operation $\odot$, satisfying (a) closure and associativity along with (b) the existence of a special element $e \in S$ (known as identity) such that, for every $g \in S$, $g \odot e = g$, and (c) there exists an element $g^{-1}\in S$ satisfying $g \odot g^{-1} = e$. If the group also satisfies the commutativity property, then it is called {\em Abelian group}. 
The {\em generators} of a group $G$ are the set of elements $g \in G$ such that repeated application of $g$ on itself produces all the elements in the group. 
A {\em generating set} of a group is a subset of the underlying set of the group such that every element of the group can be expressed as a combination (under the group operation) of finitely many elements of the subset and their inverses. The \textit{order} of a group is the number of elements within that group. 

A {\em ring} is a set $R$ together with two binary operations, written as addition and multiplication, such that $R$ is an Abelian group under addition, $R$ is closed under multiplication, $R$ satisfies the associativity of multiplication and distributive law. A ring $R$ is said to be a {\em commutative ring} if $a\odot b = b\odot a$ for all $a,b \in R$, {where $\odot$ is the binary operation of usual multiplication}.
We say that an element $a\in R$ is a {\em unit} if there exists an element $b\in R$ (called {\em inverse}) such that $a\odot b = b\odot a = 1$. The set of all units of $R$ forms a group under multiplication and is called the {\em unit group} of $R$. 

A field is a {\em commutative ring} with identity such that the nonzero elements form a group under multiplication. If $\mathbb F$ and $\mathbb K$ with $\mathbb F\subseteq \mathbb K$ are fields, then $\mathbb K$ is called a {\em field extension} of $\mathbb F$. We will refer to the pair $\mathbb F \subseteq \mathbb K$ as the field extension $\mathbb K/\mathbb F$ and to $\mathbb F$ as the {\em base field}. We make $\mathbb K$ into an $\mathbb F$-vector space by defining scalar multiplication for $\alpha \in \mathbb F$ and $a \in \mathbb K$ as $\alpha a $.

Let $G$ be a finite multiplicative group of order $n$, identity $1$, with a listing of elements $\{g_1,g_2,\cdots, g_n\}$ and let $\{x_{g_1}, x_{g_2}, \cdots, x_{g_n}\}$ be a set of independent commuting variables indexed by the elements of $G$. Let $\mathbb F$ be a {\em field} and $(\mathbb F,+,\cdot)$ be a {\em ring} equipped with ordinary addition and multiplication operations. The set of all polynomials in the indeterminate $x$ with coefficients in $\mathbb F$  is denoted by $\mathbb F[x]$. A polynomial $a(x)$ in a {\em polynomial ring} $\mathbb F[x]$ defined over a field $\mathbb F$ is {\em irreducible} if it cannot be decomposed into two polynomials $b(x),c(x) \in \mathbb F[x]$ such that $b(x)c(x)=a(x)$.

\begin{lem}
\label{lem:powerofGenerator}
    Let $g$ be the generator of a multiplicative group $G$ of order $p$ and identity $1$. Let $k \in \mathbb Z$.
    \[
    \sum_{\ell=0}^{p-1} g^{k\ell} = \begin{cases}
        p &  \text{ if $k$ is a multiple of $p$ } \\
        0 & \text{ otherwise}
    \end{cases}.
    \]
\end{lem}

\begin{proof}

    First, consider the case when $k<0$ and $k\ne 0 \mod p$. Since $g$ is a generator of the group $G$ with identity $1$, then $g^p=1$. Set $j=-k >0$. Then  
    \[
    \sum_{\ell=0}^{p-1} g^{k\ell} = {1 - g^{kp} \over 1 - g^k} = {g^{jp}- 1  \over g^{jp}(1 -g^k)} = {(g^{p})^j- 1  \over g^{jp}(1 -\omega^k)}  = 0.
    \]
    and $g$ is a generator of the group; therefore $g^p=1$. Similarly, for $k>0$  and $k\ne 0 \mod p$, we have 
    \[
    \sum_{\ell=0}^{p-1} g^{k\ell} = {1 - g^{kp} \over 1 - g^k} = 0.
    \]    

When $k=0 \mod p$, then $g^k=1$ and the result follows.
\end{proof}

\subsection{Matrices.}
The vector space of complex $n \times m$ matrices is denoted by $\complex^{n \times m}$. The set of real $n \times m$ matrices form a subspace of $\complex^{n \times m}$ and is denoted $\real^{n \times m}$. For a matrix $M$, its $(i,j)$-th entry is denoted by $M[i,j]$, the $i$-th row is denoted $M[i;]$, the $j$-th column is denoted $M[;j]$, and its complex conjugate is denoted by $M^*$.
We use the notation 
$\mathbb I_n$ to denote the $n \times n$ identity matrix,
and $0^{n \times m}$ to denote an $n \times m$ all zero matrix. 

A matrix $U$ is {\it unitary} if it satisfies $UU^* = \I_n$.
%, where $\I_n$ is the $n \times n$ identity matrix. 
The set of unitary matrices is denoted $\unitary{\complex^n}$. The eigenvalues of a unitary matrix lie on the unit circle in a complex plane. In other words, every singular value of a unitary matrix is $1$. 

For two matrices $A,B \in \complex^{m\times n}$,  the {\em Kronecker product} is 
\[
{A}\otimes {B} = \begin{bmatrix}
  A[1,1] {B} & \cdots & A[1,n]{B} \\
             \vdots & \ddots &           \vdots \\
  A[m,1] {B} & \cdots & A[m,n] {B}
\end{bmatrix}.
\]

\subsubsection{Matrix norms.}
We begin with defining matrix norms induced by vector norms. For a matrix $M \in\complex^{n \times m}$, the norm $\norm{M}_{p\rightarrow q}$ is defined as
\[
\norm{M}_{p\rightarrow q} = \max_{x\in \complex^m}\left\{\frac{\norm{Mx}_q}{\norm{x}_p}\right\}.
\]

Of particular interests are the norms $\norm{M}_{1\rightarrow 2}$ and $\norm{M}_{2\to\infty} $, which are the maximum of the $2$-norm of the columns of $M$ and the maximum of the $2$-norm of the rows of $M$, respectively.

\subsubsection{Group Pattern Matrices} 
\label{sec:groupMatrices}
Our result relies on the {\em group-pattern matrix}.

\begin{defn}
[Group-pattern matrix~\cite{dedekiend1882}]
\label{def:groupPattern}
Let {$\mathbb F$ be a field} and $M \in \mathbb F^{n \times m}$ be an $n \times m$ matrix. Let $G$ be a multiplicative 
group of order $p = \max\{n,m\}$. Then, $M$ is a group-pattern matrix for $G$ and the list $(g_1, \cdots, g_p)$ of the elements in G if and only if there is a function $h:G \to \mathbb F$ such that its $(i,j)$-entry, $M[i,j] = h(g_i^{-1}g_j).$
\end{defn}

\begin{theorem}
    [Chalkley~\cite{chalkley1976}]
    \label{thm:chalkley}
    Let $G$ be a cyclic multiplicative group of order $p$. Let $g$ be its generator. Let $B$ and $C$ be $n \times n$ matrices whose entries are elements in a commutative ring $R$. Let $B$ be a group pattern matrix for $G$ with function $b:G \to R$ and let $C$ be a group pattern matrix for $G$ with function $c:G \to R$. Then $A=BC$ is a group pattern matrix for $G$ with function 
    \begin{align}
    a(x) = \sum_{k=1}^n b(g^k) c(g^{-k} x).
    \label{eq:multiplicationChalkley}        
    \end{align}
\end{theorem}

In our setting we use $\mathbb{R}$ as field $\mathbb F$ and as commutative ring $R$ and the set of $2\streamlength$-th roots of unit as cyclic multiplicative group $G$ of order $2\streamlength$.

\subsubsection{Factorization Norms.}
One important class of matrix norm that has been extensively studied in functional analysis, operator algebra, and computer science are properties of {\em factorization norms}. 
Note first the monotonicity of each $\gamma_{\op{(p)}}$ norm:
\begin{fact}[Monotonicity of $\gamma_{(p)}(\cdot)$~\cite{haagerup1980decomposition}]
\label{claim:monotonicityGammap}
\label{claim:monotonicity}
    For any $\streamlength \in \mathbb N_+$ and $p\in [2,\infty)$, let $M$ be the matrix formed by the $\streamlength \times \streamlength$ principal submatrix of an $(\streamlength+1) \times (\streamlength+1)$ matrix $\widehat M$. Then $\gamma_{(p)}(M) \leq \gamma_{(p)}(\widehat M)$. 
\end{fact}

The first lower bound on $\gamma_{\op{2}}(M_\counting)$ was shown by Kwapien and Pelczynski~\cite{kwapien1970main}. Davidson~\cite{davidson1988} and 
Mathias~\cite{mathias1993hadamard} subsequently improved it with the following bound shown in Mathias~\cite{mathias1993hadamard}:
\begin{theorem}
    [Corollary 3.5 in Mathias~\cite{mathias1993hadamard}]
    \label{thm:mathias1993hadamard}
    Let $M_\counting$ be an $n\times n$ lower-triangular matrix with all ones. Then 
    \[
    \gamma_{\op{2}}(M_\counting) \geq \paren{\streamlength+1 \over 2\streamlength} \sum_{i=1}^\streamlength \left| \csc \paren{(2i-1)\pi \over 2\streamlength} \right|.
    \]
   
\end{theorem}

This was improved in certain regimes by Matousek et al. \cite{matouvsek2020factorization}:
\begin{align}
    \label{eq:matousek2020}
    \gamma_{\op 2}(M_\counting) \geq {\ln((2\streamlength+1)/3) + 2 \over \pi}.
\end{align}

\subsection{Differential Privacy.}
The privacy and utility guarantees studied in this paper depend on the Gaussian distribution. Given a random variable $X$, we denote by $X \sim N(\mu, \sigma^2)$ the fact that $X$ has Gaussian distribution with mean $\mu$ and variance $\sigma^2$ with the probability density function $p_X(x) = \frac{1}{\sqrt{2 \pi \sigma}} e^{-\frac{(x-\mu)^2}{2\sigma^2}}.$ 
The multivariate Gaussian distribution is the multi-dimensional generalization of the Gaussian distribution. For a random variable $X$, we denote by $X \sim N(\mu, \Sigma)$ the fact that $X$ has a multivariate Gaussian distribution with mean $\mu \in \real^d$ and covariance matrix $\Sigma \in \real^{d \times d}$ which is defined as $\Sigma = \E[(X-\mu)(X-\mu)^\top]$. The  probability density function of a multivariate Gaussian has a closed form formula:
\[
p_X(x) = \frac{1}{\sqrt{(2 \pi)^\streamlength \mathsf{det}(\Sigma)}} e^{-(x-\mu)^\top \Sigma^{-1} (x-\mu) },
\] 
where $\mathsf{det}(\Sigma)$ denotes the determinant of $\Sigma$. The covariance matrix is a positive definite matrix. 
We use the following fact regarding the multivariate Gaussian distribution:
\begin{fact}
\label{fact:multivariate}
Let $X \sim N(\mu, \Sigma)$ be a $d$-dimensional multivariate Gaussian distribution. If $M \in \complex^{\streamlength \times d}$, then the multivariate random variable $Y = MX$ is distributed as though $Y \sim N(M\mu, M\Sigma M^*)$.
\end{fact}

Our algorithm for continual counting uses the Gaussian mechanism. To define it, we need to first define the notion of $\ell_2$-sensitivity. For a function $f : \mathcal X^n \to \R^d$  its {\em $\ell_2$-sensitivity} is defined as 
\begin{align}
\Delta f := \max_{\text{neighboring }X,X' \in \calX^n} \norm{f(X) - f(X')}_2.
\label{eq:ell_2sensitivity}    
\end{align}

\begin{definition}
[Gaussian mechanism~\cite{dwork2006our}]
\label{def:gaussianmechanism}
Let $f : \mathcal X^n \to \R^d$ be a function with $\ell_2$-sensitivity $\Delta f$ and $\I_{d\times d}$ denote the $d\times d$ identity matrix. For a given $\epsilon,\delta \in (0,1)$  given $X \in \mathcal X^n$
the Gaussian mechanism $\mathsf A_{\mathsf{gauss}}$
returns $\mathsf A_{\mathsf{gauss}}(X) =  f(X) + e$, where $e \sim N(0,\sigma_{\epsilon,\delta}^2 (\Delta f)^2 \I_{d\times d})$ for  $\sigma_{\epsilon,\delta} =  \frac{2}{\epsilon} \sqrt{{\frac{4}{9} +  \ln \paren{\frac{1}{\delta}\sqrt{\frac{2}{\pi}}}}}$.
\end{definition}

\begin{theorem}
\label{thm:gaussian}
For a given $\epsilon,\delta \in (0,1)$ the Gaussian mechanism $\mathsf A_{\mathsf{gauss}}$ satisfies $(\epsilon,\delta)$-differential privacy.
\end{theorem}

In the continual observation model, we can consider different settings depending on how the stream is generated. If the entire stream is generated in advance, we call it the {\em non-adaptive} setting. A stronger setting, known as {\em adaptive} setting, is when the streamed input at $t$ depends not only on the input prior to the time $t$ but also on all the outputs given by the algorithm so far. The following result gives the reason  why we only care about non-adaptive input streams when using Gaussian mechanism:
\begin{theorem}
[Theorem 2.1 in Denisov, McMahan, Rush, Smith, and Thakurta~\cite{mcmahan2022private}]
\label{thm:denisovadaptive}
Let $A \in \real^{n\times n}$ be a lower-triangular full-rank query matrix, and let $M = LR$ be any factorization with the following property: for any two neighboring streams of vectors $x, x' \in \real^{n}$ , we have $\norm{R(x-x')} \leq \psi$. 
Let $z \sim N(0,\zeta^2 C_{\epsilon,\delta}^2 \I_{n \times n})$ with $\psi$
large enough so that $\dpmechanism_{\mathsf{matrix}}(x)=Mx+Lz=L(Rx+z)$ satisfies $(\epsilon,\delta)$-DP in the nonadaptive continual release model. Then, $\dpmechanism_{\mathsf{matrix}}$ satisfies the same DP guarantee (with the same parameters) even when the rows of the input sequence are chosen adaptively.
\end{theorem}

\section{Our Techniques}
\label{sec:techniques}
Prior works that gave explicit factorization~\cite{dvijotham2024efficient,fichtenberger2023constant,henzinger2023almost,henzinger2024unifying} use the fact that the matrix in the (weighted) continual sum problem is a lower-triangular {\em Toeplitz matrix}, i.e., a matrix $\decayingmatrix$  as stated in \cref{def:Mf}. For the decaying sum problem with weight function $f:\mathbb N \to \real_+$,  one can represent the corresponding Toeplitz matrix by the unique polynomial defined in \cref{eq:decaypolynomialHenzingerUU24}. This polynomial is a degree-$(n-1)$ polynomial. The idea in these works is to extend this polynomial to an infinite degree polynomial, i.e., a \textit{generating function}:
\begin{align}
\label{eq:generatingFunction}
m_f^{(\infty)}(x) := f(0) + f(1) x + f(2)x^2 + \cdots .
\end{align}

This function corresponds to a Toeplitz operator, $\mathsf M_f$, and the decaying sum matrix, $\decayingmatrix=\mathsf M_f[:n,:n]$, is the $\streamlength \times \streamlength$ principal submatrix of $\mathsf M_f$. They then use the fact that, for two Toeplitz operators $\mathsf P$ and $\mathsf Q$ with associated polynomials $p(x)$ and $q(x)$, respectively, their product $\mathsf P\mathsf Q$ has the associate polynomial $p(x)q(x)$~\cite{bottcher2000toeplitz}. Thus, they compute the factorization by using the first $n$ coefficients $r_{f,0}, r_{f,1},\cdots, r_{f,n-1}$ of the  polynomial $r_f^{(\infty)}(x):=(m_f^{(\infty)}(x))^{1/2} = r_{f,0} +r_{f,1} x + \cdots $.

For the decaying sum problem, Henzinger, Upadhyay, and Upadhyay~\cite{henzinger2024unifying} used Fa\`{a} di Bruno's formula~\cite{roman1980formula} to compute the coefficients of the polynomial $r_f^{(\infty)}(x)$. As a consequence, their coefficients depend on evaluating (the difficult to estimate) Bell's polynomial, and this is unavoidable because of the combinatorial structure in Faa di Bruno's formula~\cite{roman1980formula}. The second issue with this approach is that, in the end, they truncate the generating function of the square root to get the desired finite-dimensional matrix. This results in lower-triangular Toeplitz factors $L$ and $R$ of $\decayingmatrix$, which we show are not generally optimal even though they are optimal for factors restricted to lower-triangular Toeplitz matrices~\cite{dvijotham2024efficient}. 

\subsection{Our Techniques and Construction.}
We start with a simple question: \textit{do the factors need to be lower-triangular Toeplitz matrices?}  
%To understand this question, 
Let's consider the simplest example of continual counting, where $m_f(x)$ is the {\em Cyclotomic polynomial} and $m_f^{(\infty)}(x)=(1-x)^{-1}$. Therefore, $r_f^{(\infty)}(x) = (1-x)^{-1/2}$, i.e., the entries of the Toeplitz factors are the generalized binomial coefficients. 
Dvijotham, McMahan, Pillutla, Steinke, and Thakurta~\cite{dvijotham2024efficient} showed that we can approximate the degree-$(n-1)$ truncation of $(1-x)^{-1/2}$ by an $O(\ln^2(n))$-degree polynomial over $\set{e^{\iota \theta}: \theta \in [0,2\pi]}$. They also showed that if we restrict the minimization problem in \cref{eq:generalfactorizationnorm} to lower-triangular Toeplitz matrices, then the square root factorization in~\cite{fichtenberger2023constant} is an optimal factorization in terms of the $\gamma_2$-norm, i.e., it achieves the smallest upper bound of the $\gamma_2$ value over all such matrices. That is, we need a different approach to find factors that achieve a better upper bound on the $\gamma_2$ value. 
Our first observation is that for the continual (weighted) sum problem, we do not require the matrices to be either Toeplitz or square. 
In fact, the factors in the binary mechanism are neither Toeplitz nor square~\cite{henzinger2023almost}.

Using Denisov et al.~\cite{mcmahan2022private}, we have that, for any matrix $M$, there is an \emph{optimal} factorization  $LR=M$ with lower-triangular matrices $L$ and $R$ achieving $\gamma_{\op 2}(M)=\rownorm{L}\colnorm{R}$. 
Therefore, our approach is to find a lower-triangular matrix $L$ that results in the optimal factorization and then set $R$ to be $L^\dagger \decayingmatrix$, where $L^\dagger$ is the {\em Moore-Penrose pseudoinverse}. The proof in Denisov et al.~\cite{mcmahan2022private} relies on the fact that we know an optimal factorization; that is, it only proves the existence of a lower triangular factor $L$, but not a bound on $\gamma_{\op 2}$.  Theoretically, one can compute these factorizations by solving a convex program, but solving such a program with even the state-of-the-art optimizer is impractical in settings like private learning or statistics where continual counting is used as $n$ can range from $1$ million to $10$ billion. We, therefore, aim to find an algebraic approach. 

\subsubsection{Mapping to polynomials.} 
\label{sec:mappingtoPolynomial}
As in previous works~\cite{dvijotham2024efficient, fichtenberger2023constant,henzinger2024unifying}, we work with polynomials. Continuing on using the continual counting matrix as an example, recall that previous explicit factorizations~\cite{fichtenberger2023constant} considered the generating function corresponding to the continual counting matrix $M_\counting$ because the {\em cyclotomic polynomial} (or the polynomial $m_f(x)$ for continual counting) is {\em irreducible} when $\streamlength$ is a finite prime number. This issue is also the fundamental reason why prior work cannot factorize the matrix corresponding to the sliding window model~\cite{henzinger2024unifying}. Our key observation is that we do not need to find the square root of the generating function but it suffices to find another polynomial such that its coefficients for the terms up to $x^{n-1}$ match that of $m_f(x)$. This brings us to our conceptual departure from prior work.
Unlike the previous works, we lift the problem to a $(2\streamlength-1)$-degree polynomial instead of to the generating function. Therefore, instead of relying on the Fa\`{a} di Bruno formula, we rely on concepts in abstract algebra developed to study the irreducibility of polynomials, ensuring that we have a nice algebraic structure instead of the more combinatorial structure of Bell's polynomial. 

Using a $(2\streamlength-1)$-degree polynomial, say $a_f(x)  \in \real[x]$ (instead of infinite degree polynomials as in previous works~\cite{dvijotham2024efficient, fichtenberger2023constant,henzinger2024unifying}), for a given weighted sum problem defined by a function $f: \mathbb N \to \real$ introduces several challenges, and tackling them leads to our approach: 

\begin{enumerate}
    \item The square root of  $a_f(x)$ is no longer in $\real[x]$ because it is an odd degree polynomial.     
    \item We have to either ensure that the polynomial $a_f(x)$ is not irreducible or find another technique. For example, if for $M_\counting$, i.e., for the constant function $f(n)=1$, we restrict the generating function $(1-x)^{-1}$ to the degree-$(2\streamlength-1)$ polynomial, we still have a cyclotomic polynomial that is irreducible when $2\streamlength-1$ is a prime number~\cite{weintraub2013several}.\label{item:irreducibility}    
    \item Finally, even if our constructed polynomial $a_f(x)$ is reducible and we can find two factors, we still have to ensure that the resulting factor polynomials do not have any coefficient with a large absolute value. If not, this would lead to a weak bound on the factorization norm because the bound is the sum of the square of coefficients of these factors. This is the reason why \cite{henzinger2024unifying} only considered monotonic weight functions. 
    \label{eq:largeCoefficient}
\end{enumerate}

One approach can be to define a polynomial $r_f(x) \in \complex[x]$  such that the restriction of $(r_f(x))^2 $ to a degree $(2\streamlength-1)$-polynomial is exactly $a_f(x)$ and all the coefficients of $r_f$ have bounded absolute value. We can then take the first $n$ coefficients of $r_f(x)$. The question then is how to define such a polynomial without ending up in a situation where we have lower-triangular Toeplitz factors and, consequently, no improvement over prior bounds.  

Our approach revisits this problem from a group-theoretic perspective, more specifically, the {\em group matrix} introduced by Dedekind~\cite{dedekiend1882} (also see \Cref{def:groupPattern}). In particular, we define a {cyclic multiplicative group} $G$ and construct a polynomial $a_f:G \to \real$ such that $\decayingmatrix$ is a submatrix of a group-pattern matrix $\widehat \decayingmatrix \in \complex^{2\streamlength \times 2\streamlength}$ for $G$. We then apply Theorem~\ref{thm:chalkley} to $\widehat \decayingmatrix$ to find a factorization of $\widehat \decayingmatrix$, which in turn leads to a novel factorization of $\decayingmatrix$. 
Note that, in  Theorem~\ref{thm:chalkley}, we do not require $a(x)$ to be reducible because only $c(\cdot)$ is a function of $x$ and $b(g^k)$ does not depend on $x$. As a result, we do not have to worry about irreducibility of $a_f(x)$, resolving the issue mentioned in \cref{item:irreducibility} on page~\pageref{item:irreducibility}.

We use a group that simplifies the calculations of the $\gamma_{\op 2}$-norm for specific matrices like $M_\counting$. Towards this goal, note that the set of roots of unity forms a {\em cyclic multiplicative group}. In particular, let $S:=\set{\omega_1,\omega_1, \cdots, \omega_{2\streamlength}}$ denote the set of $2\streamlength$-th roots of unity, i.e., $\omega_i = \omega^i$ for the generator $\omega = e^{\iota \pi/\streamlength}$. Then $G=(S,\times)$ is a cyclic group of order $2\streamlength$ with \textit{generator} $\omega$
%=e^{\pi\iota/\streamlength}$ 
under the binary operation of complex multiplication, $\times$. We construct the following polynomial: 
\begin{align}
\label{eq:polynomialrepresentationFinal}
a_f(x) = {1 \over 2\streamlength} \sum_{\ell=0}^{2\streamlength-1}  \paren{\sum_{k=0}^{n-1}f(k)\omega^{k\ell}} x^\ell = {1 \over 2\streamlength}\sum_{\ell=0}^{2\streamlength-1} m_f(\omega^\ell) x^\ell .    
\end{align}

We chose $a_f$ such that $a_f(\omega^{-d})$ for integer $-n \le d \le 2\streamlength-1$ fulfills the following condition, which we will show to be very useful:
\begin{align}
a_f(\omega^{-d}) = \begin{cases}
    0 & -n \leq d \leq -1  \\
    f(d) & 0 \leq d \leq n-1 \\
    0 & \streamlength \leq d \leq 2\streamlength-1
\end{cases}. 
\label{eq:groupPatterndecayingmatrix}
\end{align}

Now we denote  $\widehat \decayingmatrix$ to be the group-pattern matrix for $G$,  $(g_1, \dots, g_{2\streamlength})$, and $h$ 
in Definition \ref{def:groupPattern}, where $G=(S,\times)$ is the above cyclic group of roots of unity, $g_i = \omega_i$, and $h= a_f$.
Thus $\widehat \decayingmatrix [i,j] = a_f(\omega_i^{-1}\omega_j)$ for $1 \le i,j \le 2\streamlength$.
Now recall that $M_f[i,j]=f(i-j)$ when $i\geq j$, and $0$, otherwise. \textit{The crucial observation is now that  $M_f[i,j]=a_f(\omega^j \omega^{-i})$ for $1\leq i,j \leq \streamlength$, i.e.,  $M_f$ is an $\streamlength \times \streamlength$ principal submatrix of $\widehat \decayingmatrix$. } As a consequence, a factorization of $\widehat \decayingmatrix$ leads to a factorization of $\decayingmatrix$. 

Our approach is very general: The group $G$ in our approach needs to be cyclic, multiplicative, and of order $2\streamlength$, but it is otherwise unrestricted. Thus, different choices for $G$ can lead to different group-pattern matrices, which in turn can lead to different factorizations. In particular, depending on the application, one might choose different groups that allow fast computation, low-space representation, etc. For example, performing computation on the extension of the binary field, $\mathbb F_{2^m}$ is more efficient than on a prime field, $\mathbb F_q$ or its extension.

One can contrast this with the ``restrictive" operator action used in earlier work~\cite{henzinger2024unifying}: (a) the polynomial is the generating function, $m_f^{(\infty)}(x)$; (b) the linear operator is $D^{(i-j)}$ for $i\geq j$, where $D$ is the differential operator; (c) the underlying group is the cyclic additive group with identity $0$; and (d) the polynomial evaluation is done at the identity element. That is, their evaluation point and operator were restricted once the polynomial $m_f^{(\infty)}(x)$ was fixed. Removing these restrictions is the intuitive reason for our improvements and simplification.

\subsubsection{Finding a factorization and bound on \texorpdfstring{$\gamma_{\op 2}(\decayingmatrix)$}{Lg}.}
Now that we have defined an appropriate polynomial, we need to compute its factors.
Unlike previous factorizations~\cite{dvijotham2024efficient, fichtenberger2023constant, henzinger2023almost, henzinger2024unifying}, we do not restrict ourselves to computing a factor such that $L$ and $R$ are real-valued lower-triangular matrices in this subsection; instead we  first compute factors and then show in Section \ref{sec:removingComplex} how to convert them into real-valued and lower-triangular factors using standard tricks.
Our factorization follows from basic concepts in group-matrix theory: if we consider the coefficients of polynomial $a_f(x)$ to form the entries of the matrices (shown in \cref{eq:groupPatterndecayingmatrix}), then by Theorem \ref{thm:chalkley}, it suffices to find polynomials $b_f(x), c_f(x)$ such that
\begin{align}
\label{eq:complexFactorization}
\sum_{k=0}^{2\streamlength-1} b_f(\omega^k) c_f(\omega^{-k}x) = a_f(x). 
\end{align}

Depending on the use case, such as accuracy or efficiency with respect to space and time, we can now find different factorizations, one of which can be that $b_f(x)=c_f(x)$. In \Cref{sec:proofComplexFactorization}, we show that \cref{eq:complexFactorization} is satisfied by using 
\begin{align}
b_f(x) = c_f(x) = {1 \over 2\streamlength} \sum_{\ell=0}^{2\streamlength-1} x^\ell \paren{\sum_{k=0}^{n-1}f(k)\omega^{k\ell}}^{1/2} = {1 \over 2\streamlength}\sum_{\ell=0}^{2\streamlength-1} \paren{m_f(\omega^\ell)}^{1/2}x^\ell.
\label{eq:polynomialfactorFinal}    
\end{align}
 
As before, let $G$ be the cyclic multiplicative group of roots of unity of order $2\streamlength$ under complex multiplication and let $\omega=e^{i\pi/n}$.
We now define matrices $\widetilde L \in \complex^{\streamlength \times 2\streamlength}$ and $\widetilde R \in \complex^{2\streamlength \times \streamlength}$ as follows:  
\begin{align}
    \label{eq:widetildeL}
    \text{for all} \quad 1 \le i \leq  \streamlength, 1 \leq j \le 2\streamlength, \quad  \widetilde L[i,j] :=b_f(\omega^{j-i})
\end{align} to be the group pattern matrix over $G$ 
and the list $(\omega^{2\streamlength}=1,\omega, \omega^2, \dots, \omega^{2\streamlength-1})$, 
and set $\widetilde R = \widetilde L^*$, where $L^*$ is the complex conjugate of $L$. 

The fact that $\widetilde R$ is a group pattern matrix with function $c_f(x)$ follows from the fact that $\omega_k^*=\omega^{-k}=\omega^{2\streamlength-k}$ for all $1\leq k \leq 2\streamlength$. In particular, for all $1 \leq i \leq 2 \streamlength$ and $1\leq j \leq \streamlength$, 
\begin{align*}
    \widetilde R[i,j] &= \widetilde L^*[i,j] = (\widetilde L[j,i])^* = b_f(\omega^{i-j})^* = b_f(\omega^{2\streamlength-(i-j)}) = b_f(\omega^{j-i}) = c_f(\omega^{j-i}),
\end{align*}
where we used the fact that $f$ is real-valued.

Equation~\ref{eq:polynomialfactorFinal} allows us to bound $\ptrace{\widetilde L}\colnorm{\widetilde R}$. 
In particular, let $N=\widetilde L\widetilde L^* \in \complex^{\streamlength \times \streamlength}$.
Using the fact that $\widetilde R = \widetilde L^*$, we show in \Cref{sec:proofMainUpperBoundGamma} (see  \Cref{claim:traceBoundFirststep}) the following:
\begin{align*}
    \colnorm{\widetilde R}^2  \leq {1 \over 2\streamlength} \sum_{\ell=0}^{2\streamlength-1} \left|m_f(\omega^\ell)\right| \quad \text{and} \quad N[i,i] \leq {1 \over 2\streamlength} \sum_{\ell=0}^{2\streamlength-1} \left|m_f(\omega^\ell)\right| \quad \text{for}~1 \leq i \leq 2\streamlength.
\end{align*}

This immediately implies that 
 % as we show below. 
\begin{align}
    \ptrace{\widetilde L} = \paren{\sum_{i=1}^{\streamlength} N[i,i]^{p/2} }^{1/p} \leq   \paren{n \left( {1 \over 2\streamlength} \sum_{\ell=0}^{2\streamlength-1}  \left|m_f(\omega^\ell)\right| \right) ^{p/2} }^{1/p} \le n^{1/p} \paren{{1 \over 2\streamlength} \sum_{\ell=0}^{2\streamlength-1} \left|m_f(\omega^\ell)\right|}^{1/2}.
\end{align}

From  Theorem \ref{thm:chalkley} it follows that $\widetilde L \widetilde R^* \in  \real^{\streamlength \times \streamlength}$ is a group pattern matrix for $G$ with function $a_f(x)$. But by our observation above $\decayingmatrix$ has exactly the same entries as this group pattern matrix, i.e., it follows that 
%From \cref{eq:complexFactorization} and Theorem \ref{thm:chalkley},
$\decayingmatrix=\widetilde L \widetilde R$. Theorem \ref{thm:mainupperboundgamma} now follows:
\[
\gamma_{\op{(p)}}(M_f)  \leq \colnorm{\widetilde R} \ptrace{\widetilde L} \leq  { n^{1/p} \over 2\streamlength} \sum_{\ell=0}^{2\streamlength-1} \left|m_f(\omega^\ell)\right|.
\]

Note that the only missing pieces of this analysis are 
(1) a proof of~\cref{eq:groupPatterndecayingmatrix}, (2) a proof that \cref{eq:complexFactorization} is satisfied with our choice of $b_f$ and $c_f$, (3) the upper bounds on $\colnorm{\widetilde R}$ and $N(i,i)$.
We will first prove these missing pieces in \Cref{sec:proofgroupPatterndecayingmatrix} and \Cref{sec:proofComplexFactorization} before discussing how to turn our new factorization into a new differentially private algorithm for continual release in \Cref{sec:removingComplex}.

\subsection{Proof of \texorpdfstring{\cref{eq:groupPatterndecayingmatrix}}{Lg}.} 
\label{sec:proofgroupPatterndecayingmatrix}
When $x=\omega^{-d}$ for $0 \leq d \leq n-1$,  \Cref{lem:powerofGenerator} gives 
\begin{align*}
a_f(\omega^{-d}) 
    &= {1 \over 2\streamlength} \sum_{\ell=0}^{2\streamlength-1}\paren{1 + f(1) \omega^\ell + f(2) \omega^{2\ell} + \cdots + f(n-1)\omega^{(n-1)\ell}} \omega^{-d\ell} \\ 
    &= {1 \over 2\streamlength} \paren{ \sum_{\ell=0}^{2\streamlength-1} f(d) +   \underbrace{\sum_{k=0}^{d-1} f(k) \sum_{\ell=0}^{2\streamlength-1} \omega^{(k-d)\ell} + \sum_{k=d+1}^{\streamlength-1} f(k) \sum_{\ell=0}^{2\streamlength-1} \omega^{(k-d)\ell}}_{=0 \text{ by setting $g=\omega,p=2\streamlength$ in \Cref{lem:powerofGenerator}}} } 
    = f(d).
\end{align*}

For the other two cases, recall that $1,\omega, \cdots, \omega^{2\streamlength-1}$ are roots of polynomial $z^{2\streamlength}-1$. Therefore, for all  integer $k\geq 1$ that is not a multiple of $2\streamlength$, 
$$ \sum_{\ell=0}^{2\streamlength-1} \omega^{k \ell} =
\sum_{\ell=0}^{2\streamlength-1} \omega_k^\ell =\frac{\omega_k^{2\streamlength}-1}{\omega_k-1} = 0,$$
where $\omega_k = \omega^k$ is a roots of the polynomial $z^{2\streamlength}-1$.  
This implies that, for $x=\omega^{-d}$ when $\streamlength \leq d \leq 2\streamlength-1$ and $-\streamlength \leq d \leq -1$, $a_f(\omega^{-d}) = 0$. In more detail, when $-n \leq d \leq -1$, note that $k-d \ge 1$ for all integer $k\ge 0$. Then we have 
\begin{align*}
a_f(\omega^{-d}) 
    &= {1 \over 2\streamlength} \sum_{\ell=0}^{2\streamlength-1} \omega^{d \ell} \sum_{k=0}^{n-1}f(k) \omega^{k \ell}= {1 \over 2\streamlength} \sum_{k=0}^{n-1}f(k) \sum_{\ell=0}^{2\streamlength-1} \omega^{(k-d)\ell}=0.
\end{align*}

When $\streamlength \leq d \leq 2\streamlength-1$, then $k-d \leq -1$ (or $1 \leq d-k \leq 2\streamlength-1$) as $0 \leq k \leq \streamlength-1$. Let $c=d-k$. Then 
\begin{align*}
\sum_{\ell=0}^{2\streamlength-1} \omega^{(k-d)\ell} 
    &= {1 \over \omega^{2(k-d)}}\sum_{\ell=0}^{2\streamlength-1} \omega^{(d-k)\ell} = {1 \over \omega^{-2c}}\sum_{\ell=0}^{2\streamlength-1} \omega^{c\ell} = {1 \over \omega^{2(k-d)}} \times {\omega^{2cn}-1 \over \omega^c -1 } =  0 
\end{align*}
because $\omega^c$ is a root of polynomial $z^{2\streamlength}-1$. Therefore, as in the case of $-n \leq d \leq 1$, 
we have that $a_f(\omega^{-d})=0$. 

Note that we would have arrived at the same derivation had we picked a different group and invoked Lemma~\ref{lem:powerofGenerator} instead of the properties of the roots of $z^{2\streamlength}-1$.

\subsection{Proof of \texorpdfstring{\cref{eq:complexFactorization}}{Lg}.}
\label{sec:proofComplexFactorization}
Let  $\zeta_\ell = \sqrt{m_f(\omega^\ell)}$ for brevity below. We prove  \cref{eq:complexFactorization} for 
\[
b_f(x) = c_f(x) = {1 \over 2\streamlength} \sum_{\ell=0}^{2\streamlength-1} x^\ell \paren{\sum_{k=0}^{n-1}f(k)\omega^{k\ell}}^{1/2} = {1 \over 2\streamlength}\sum_{\ell=0}^{2\streamlength-1} \sqrt{m_f(\omega^\ell)  }x^\ell
= 
 {1 \over 2\streamlength}\sum_{\ell=0}^{2\streamlength-1} \zeta_\ell x^\ell
.
\]

Recalling the definition of $a_f(x)$ in \cref{eq:polynomialrepresentationFinal}, we have \cref{eq:complexFactorization} as follows:
\begin{align*}
    \sum_{k=0}^{2\streamlength-1} b_f(\omega^k)c_f(\omega^{-k}x) & 
    = {1 \over 4\streamlength^2}  \sum_{k=0}^{2\streamlength-1} \paren{\sum_{\ell=0}^{2\streamlength-1} \zeta_\ell \omega^{k\ell}} \paren{\sum_{\ell=0}^{2\streamlength-1} \zeta_\ell x^\ell\omega^{-k\ell}} 
     = {1 \over 4\streamlength^2}\sum_{k=0}^{2\streamlength-1} \sum_{j=0}^{2\streamlength-1} \sum_{\ell=0}^{2\streamlength-1} \zeta_\ell \zeta_j \omega^{-k\ell} x^\ell \omega^{jk} \\ 
    &= {1 \over 4\streamlength^2}\sum_{\ell=0}^{2\streamlength-1} x^\ell \zeta_\ell \sum_{k=0}^{2\streamlength-1}\sum_{j=0}^{2\streamlength-1}  \zeta_j \omega^{-k\ell}  \omega^{jk} 
     = {1 \over 4\streamlength^2}\sum_{\ell=0}^{2\streamlength-1} x^\ell \zeta_\ell \sum_{k=0}^{2\streamlength-1}\sum_{j=0}^{2\streamlength-1}  \zeta_j  \omega^{(j-\ell)k} \\
    &= {1 \over 4\streamlength^2}\sum_{\ell=0}^{2\streamlength-1} x^\ell \zeta_\ell \paren{ \sum_{k=0}^{2\streamlength-1} \zeta_\ell  + \underbrace{\sum_{j=0}^{\ell-1}\zeta_j \sum_{k=0}^{2\streamlength-1}   \omega^{(j-\ell)k} + \sum_{j=\ell+1}^{2\streamlength-1}\zeta_j \sum_{k=0}^{2\streamlength-1}   \omega^{(j-\ell)k} }_{=0 \text{ by setting $g=\omega,p=2\streamlength$ in \Cref{lem:powerofGenerator}}} } \\
    &= {1 \over 4\streamlength^2}\sum_{\ell=0}^{2\streamlength-1} x^\ell \zeta_\ell \sum_{k=0}^{2\streamlength-1} \zeta_\ell  
    = {1 \over 2\streamlength} \sum_{\ell=0}^{2\streamlength-1} \zeta^2_\ell x^\ell  
     = a_f(x),
\end{align*}
where the last equality used the definition of $\zeta_\ell$.

\subsection{Proof of \texorpdfstring{Theorem~\ref{thm:mainupperboundgamma}}{Lg}.}
\label{sec:proofMainUpperBoundGamma}
We will now show that the algorithm  {\scshape Compute-Factor} in \Cref{alg:factorization} fulfills the requirements of the theorem.
Let $\zeta_k = \sqrt{m_f(\omega^k)}$, so $\zeta_k \zeta_k^* = |m_f(\omega^k)|$. As $\widetilde L = \widetilde R^*$, $\norm{\widetilde L}_{2 \to \infty} =  \norm{\widetilde R}_{1 \to 2}$, and, thus, it suffices to upper bound $\rownorm{\widetilde L}$.   
Recall that $\widetilde L$ is a group-pattern matrix for the group $G$ with the function $b_f(x)$. Therefore,  
$\widetilde L[i,j] = b_f(\omega^{j-i}).$ 
Recall that  $\decayingmatrix= \widetilde L \widetilde R^*$. %Using the fact that $\widetilde L=\widetilde R^*$ and    $\left|\zeta_\ell \right|^2 = \left|m_f(\omega^\ell)\right|$, we have 

\begin{claim}
    \label{claim:equal}
    For any $1\leq j \leq \streamlength$, we have 
    \[\norm{\widetilde R[:,j]}_2^2 = {1 \over 2\streamlength} \sum_{\ell=0}^{2\streamlength-1} |m_f(\omega^\ell)|.\]
\end{claim}
\begin{proof}
    First note that, using the cyclic property of $\omega^i$, i.e., $\omega^i=\omega^{i+2k\streamlength}$ for all integers $k$. Therefore, %we have for $1 \leq j\leq 2\streamlength$
    \begin{align*}
    \norm{\widetilde R[:j]}_2^2 &=
    \sum_{i=0}^{2\streamlength-1} |b_f(\omega^{j-i})|^2 =     {1 \over 4\streamlength^2} \sum_{i=0}^{2\streamlength-1} \left| \sum_{\ell=0}^{2\streamlength-1} \sqrt{m_f(\omega^\ell)} \omega^{(j-i)\ell)} \right|^2 = {1 \over 4\streamlength^2} \sum_{i=0}^{2\streamlength-1} \left| \sum_{\ell=0}^{2\streamlength-1} \zeta_\ell \omega^{(j-i)\ell} \right|^2 \\
    %&= {1 \over 4\streamlength^2}  \sum_{i=0}^{2\streamlength-1} \brak{\sum_{\ell=0}^{2\streamlength-1} \zeta_\ell \omega^{(j-i)\ell} , \sum_{\ell=0}^{2\streamlength-1} \zeta_\ell \omega^{(j-i)\ell} } 
    & ={1 \over 4\streamlength^2} \sum_{i=0}^{2\streamlength-1} \paren{\sum_{\ell=0}^{2\streamlength-1} \zeta_\ell \omega^{(j-i)\ell}} \paren{ \sum_{\ell=0}^{2\streamlength-1} \zeta_\ell \omega^{(j-i)\ell} }^* 
   %& = {1 \over 4\streamlength^2}   \sum_{i=0}^{2\streamlength-1} \sum_{k=0}^{2\streamlength-1}\sum_{\ell=0}^{2\streamlength-1} \zeta_k \omega^{(k-\ell)(j-i)} \zeta_\ell^* 
   = {1 \over 4\streamlength^2}   
   \sum_{k=0}^{2\streamlength-1}\sum_{\ell=0}^{2\streamlength-1} \sum_{i=0}^{2\streamlength-1}  \zeta_k\zeta_\ell^* \omega^{(k-\ell)(j-i)}  \\
   & = {1 \over 4\streamlength^2}  \sum_{i=0}^{2\streamlength-1} 
   \sum_{\ell=0}^{2\streamlength-1} \zeta_\ell  \zeta_\ell^* + {1 \over 4\streamlength^2} \max_{j\in [\streamlength]}  
   \underbrace{\sum_{k=0}^{2\streamlength-1}\sum_{\ell \neq k}^{2\streamlength-1} \zeta_k \zeta_\ell^* \omega^{j(k-\ell)} \sum_{i=0}^{2\streamlength-1} \omega^{(\ell-k)i} }_{=0 \text{ by \Cref{lem:powerofGenerator}}}   
   = {1 \over 2\streamlength} \sum_{\ell=0}^{2\streamlength-1} |m_f(\omega^\ell)|.
    \end{align*}
    This completes the proof of \Cref{claim:equal}.
\end{proof}

\begin{remark}
    \Cref{claim:equal} shows that the $\ell_2$-norm of all the columns of $\widetilde R$ are equal. Therefore, dividing  $\widetilde R$ (and multiplying $\widetilde L$ by the scalar $\paren{{1 \over 2\streamlength} \sum_{\ell=0}^{2\streamlength-1} |m_f(\omega^\ell)|}^{1/2}$ gives a column normalized right factor matrix. 
    {Note that this leaves the upper bounds for the factorization norms achieved by the resulting factorization unchanged.}
    In practice, it has been noted by several works that finding factors under the constraint that the right factor is column normalized gives much better accuracy~\cite{choquette2022multi} and they use optimization to find such a factor. Ours is the first work that finds explicitly such a factor. 
\end{remark}

\noindent \textbf{Bounding the column norm, $\colnorm{\widetilde R}$.} 
Using \Cref{claim:equal}, we  have 
\begin{align*}
   \colnorm{\widetilde R}^2 &= \max_{1 \leq j \leq \streamlength} \norm{\widetilde R[:,j]}^2 
   = {1 \over 2\streamlength} \sum_{\ell=0}^{2\streamlength-1} |m_f(\omega^\ell)|.
\end{align*}

\noindent \textbf{Bounding the column norm, $\rownorm{\widetilde R}$.} 
We now have

\begin{align*}
   \colnorm{\widetilde R}^2 &= {1 \over 4\streamlength^2} \max_{j\in [\streamlength]}  \sum_{i=0}^{2\streamlength-1} \left| \sum_{\ell=0}^{2\streamlength-1} \zeta_\ell \omega^{(j-i)\ell)} \right|^2 
   ={1 \over 4\streamlength^2} \max_{j\in [\streamlength]}  \sum_{i=0}^{2\streamlength-1} \paren{\sum_{\ell=0}^{2\streamlength-1} \zeta_\ell \omega^{(j-i)\ell)}} \paren{ \sum_{\ell=0}^{2\streamlength-1} \zeta_\ell \omega^{(j-i)\ell)} }^* \\ 
   & = {1 \over 4\streamlength^2} \max_{j\in [\streamlength]}  \sum_{i=0}^{2\streamlength-1} 
   \sum_{k=0}^{2\streamlength-1}\sum_{\ell=0}^{2\streamlength-1} \zeta_k \omega^{(k-\ell)(j-i)} \zeta_\ell^* = {1 \over 4\streamlength^2} \max_{j\in [\streamlength]}  
   \sum_{k=0}^{2\streamlength-1}\sum_{\ell=0}^{2\streamlength-1} \sum_{i=0}^{2\streamlength-1}  \zeta_k\zeta_\ell^* \omega^{(k-\ell)(j-i)}  \\
   & = {1 \over 4\streamlength^2}  \sum_{i=0}^{2\streamlength-1} 
   \sum_{\ell=0}^{2\streamlength-1} \zeta_\ell  \zeta_\ell^* + {1 \over 4\streamlength^2} \max_{j\in [\streamlength]}  
   \underbrace{\sum_{k=0}^{2\streamlength-1}\sum_{\ell \neq k}^{2\streamlength-1} \zeta_k \zeta_\ell^* \omega^{j(k-\ell)} \sum_{i=0}^{2\streamlength-1} \omega^{(\ell-k)i} }_{=0 \text{ by \Cref{lem:powerofGenerator}}}   
   = {1 \over 2\streamlength} \sum_{\ell=0}^{2\streamlength-1} |m_f(\omega^\ell)|.
\end{align*}

\noindent \textbf{Bounding the Frobenius norm, $\norm{\widetilde R}_{\op F}$.} 
We can also estimate the $\norm{\widetilde R}_{\op F}^2$ and $\norm{\widetilde L}_{\op F}^2$, which allows us to bound the $\gamma_{(2)}$ norm (and consequently the mean squared error~\cite{henzinger2023almost}) by using \Cref{claim:equal}:
\begin{align*}
    \norm{\widetilde R}_{\op F}^2 = \sum_{j=0}^{2\streamlength-1} \norm{\widetilde R[:j]}_2^2 
    = \sum_{\ell=0}^{2\streamlength-1} |m_f(\omega^\ell)|. 
\end{align*}

In order to compute $\gamma_{\op{(p)}}$ for general $p\in [2,\infty)$, we need the following claim: 
\begin{claim}
\label{claim:traceBoundFirststep}
Let $N=\widetilde L \widetilde L^*$ for $\widetilde L$ defined above. Then 
\begin{align}
\label{eq:boundpTrace}
    N[i,i] = {1 \over 2\streamlength}   \sum_{\ell=0}^{2\streamlength-1} \left|{m_f(\omega^\ell)}  \right| \quad \text{for all}\quad 1\leq i \leq 2\streamlength.
\end{align}
\end{claim}
\begin{proof}
As $\widetilde L = \widetilde R^*$ using~\Cref{claim:equal}  
we have the following:
\begin{align*}
    N[i,i] &= \norm{\widetilde L[i:]}^2_2 =  {1 \over 4\streamlength^2} \sum_{j=0}^{2\streamlength-1} \left| \sum_{\ell=0}^{2\streamlength-1} \sqrt{m_f(\omega^\ell)} \omega^{(j-i)\ell)} \right|^2    = {1 \over 2\streamlength} \sum_{\ell=0}^{2\streamlength-1} |m_f(\omega^\ell)|.
\end{align*}
which is exactly \cref{eq:boundpTrace}.
\end{proof}

Claim \ref{claim:traceBoundFirststep} immediately implies that 
\begin{align}
    \label{eq:traceBound}
    \ptrace{\widetilde L} = \paren{\sum_{i=1}^{\streamlength} N[i,i]^{p/2} }^{1/p} \leq n^{1/p}  \paren{{1 \over 2\streamlength}\sum_{\ell=0}^{2\streamlength-1} \left|m_f(\omega^\ell)\right|}^{1/2}
\end{align}

Since $L=R^*$, we have from the definition of $\gamma_{\op{(p)}}(\cdot)$ and Fact \ref{claim:monotonicityGammap},
\begin{align*}
    \gamma_{\op{(p)}}(\decayingmatrix) 
    &\leq \paren{\sum_{i=1}^\streamlength \paren{{1 \over 2\streamlength}   \sum_{\ell=0}^{2\streamlength-1} \left|{m_f(\omega^\ell)}  \right|}^{p/2}}^{1/p} \paren{{1 \over 2\streamlength} \sum_{\ell=0}^{2\streamlength-1} \left|{m_f(\omega^\ell)}  \right|}^{1/2} \\
    &= {n^{1/p} \over 2\streamlength}\paren{ \paren{\sum_{\ell=0}^{2\streamlength-1} \left|{m_f(\omega^\ell)}  \right|}^{p/2}}^{1/p} \paren{ \sum_{\ell=0}^{2\streamlength-1} \left|{m_f(\omega^\ell)}  \right|}^{1/2}
     = {n^{1/p} \over 2\streamlength}  \sum_{\ell=0}^{2\streamlength-1} \left|{m_f(\omega^\ell)}  \right|.
\end{align*}

\begin{algorithm}[t]
\caption{{\scshape Compute-Factor}}
\begin{algorithmic}[1]
   \Require A matrix $\decayingmatrix$ based on a weight function, $f:\mathbb N \to \real$.
   \Ensure Real matrices $L$ and $R$ such that $LR=\decayingmatrix$.
    \State For the weight function, $f$, define the degree-$(2\streamlength-1)$ polynomial $a_f(x) \in \complex[x]$ as in \cref{eq:polynomialrepresentationFinal} and compute a degree-$(2\streamlength-1)$ polynomial $b_f(x) \in \complex[x]$ as in \cref{eq:polynomialfactorFinal}.

    \State Let $\widetilde{L}, \widetilde{R}^*$ be the pattern matrix witnessed by $b_f(x)$ and defined in \cref{eq:widetildeL}.

    \State Define $\widetilde L_r$, $\widetilde L_c$, $\widetilde R_r$, and $\widetilde R_c$ by $\widetilde L  = \widetilde L_r + \iota \widetilde L_c  \in \complex^{\streamlength \times (2\streamlength)}$ 
    and $\widetilde R  = \widetilde R_r + \iota \widetilde R_c  \in \complex^{(2\streamlength) \times \streamlength}$.  Define $\widehat L = \begin{pmatrix}
     \widetilde L_r & \widetilde L_c   
    \end{pmatrix}  \in \real^{n \times 4n}$ and $\widehat R = \begin{pmatrix}
        \widetilde R_r & -\widetilde R_c
    \end{pmatrix}^* \in \real^{4n \times n}$ from the solution $\widetilde L$ and $\widetilde R$. 

    \State Use the trick of Denisov et al.~\cite{mcmahan2022private} (i.e., Gram-Schmidt decomposition) to compute a decomposition of $\widehat L^\top = Q^\top L^\top$ into matrices $Q^\top$ and $L^\top$ such that $Q^\top$ is orthogonal and $L^\top$ is upper triangular matrix.

    \State \textbf{Return} $L$ and $R=Q^\top \widehat R$.

   \end{algorithmic}
\label{alg:factorization}
\end{algorithm}

\section{A differentially private continual release algorithm for an adaptive adversary}
\label{sec:removingComplex}
We showed in the previous section how to decompose the group-pattern matrix $\decayingmatrix \in \mathbb C^{ \streamlength \times  \streamlength}$ into two matrices $\widetilde L \in \complex^{ \streamlength \times 2 \streamlength}$ and $\widetilde R \in \complex^{2 \streamlength \times \streamlength}$. 
Now the output $\widetilde L(\widetilde Rx+z)$  is complex-valued as $\widetilde L$ is a complex-valued matrix, and $z$ is a simple multivariate Gaussian {of dimension $2\streamlength$}.
%, making it less useful for practitioners. 
We next discuss that if $M_f$ is a real-valued matrix, we can resolve this issue without changing the upper bound on  $\gamma_2(\decayingmatrix)$, and, hence, the additive error.

\begin{algorithm}[t]
\caption{Matrix Mechanism for Continual Observation, $\mathsf A_{\op{fact}}$}
\label{alg:factorizationmechanism}
\begin{algorithmic}[1]
   \Require A stream $(x_1,\cdots, x_\streamlength)$ of length $\streamlength$, $(\epsilon,\delta)$: privacy budget, and weight function, $f:\mathbb N \to \real$.
   \Ensure A stream estimating $\sum_{i=1}^t f(t-i) x_i$ for every time epoch $1 \leq t \leq \streamlength$.
    \State $(L,R) \gets \text{{\scshape Compute-Factor}}(\decayingmatrix)$. \Comment{\Cref{alg:factorization}}

   \For {$t$ in $1, 2, \cdots, n$}
        \State Let $L_t$ denote the first $t$ columns of the $t$-th row of $L$.
        \State On getting $x_t$, sample $z \sim N\paren{0, \sigma_{\epsilon,\delta}^2 \norm{R}_{1 \to 2}^2 \mathbb I_t}$.
        \State Output $a_t = \paren{\sum_{i=1}^t f(t-i) x_i} + \ip{L_t}{z}.$
   \EndFor
   \end{algorithmic}
\end{algorithm}

\subsection{Getting real-valued factors.}
Since both $\widetilde L$ and $\widetilde R$ are complex-valued matrices, we can decompose them into real and imaginary parts:
$\widetilde L = \widetilde L_r + \iota \widetilde L_c \quad \text{and} \quad 
\widetilde R = \widetilde R_r + \iota \widetilde R_c.$ Since $\decayingmatrix$ is a real-valued matrix, comparing the real and imaginary parts implies that 
$\widetilde L_c \widetilde R_r + \widetilde L_r \widetilde R_c = 0 \quad \text{and} \quad \widetilde L_r\widetilde R_r - \widetilde L_c \widetilde R_c = \decayingmatrix.$
 Therefore, we can consider the following real-valued factorization
\begin{align}
\label{eq:widehatFactors}
\widehat L  \widehat R = \decayingmatrix, \quad \text{where} \quad \widehat L = \begin{pmatrix}
     \widetilde L_r & \widetilde L_c   
    \end{pmatrix} \quad \text{and} \quad \widehat R= \begin{pmatrix}
        \widetilde R_r \\ -\widetilde R_c
    \end{pmatrix}.
\end{align}

Note that $\widehat L$ and $\widehat R$ are real-valued matrices. Since $|a+\iota b|^2={a^2+b^2}$ for any complex number $a+\iota b \in \complex$, it is easy to verify that $\gamma_2(\decayingmatrix) \leq \rownorm{\widehat L}\colnorm{\widehat R} = \rownorm{\widetilde L} \colnorm{\widetilde R} $. Similarly, $\mathsf{Tr}_{\op p}(\widehat L) = \mathsf{Tr}_{\op p}(\widetilde L)$. This shows that we can use the real-valued factorization $\widehat L, \widehat R$ instead of the complex-valued factorization $\widetilde L, \widetilde R$, without increasing the upper bound on $\gamma_2(\decayingmatrix)$.

\subsection{Getting a lower-triangular factors.}
\label{sec:gettingLowerTriangular}
The factorization in Henzinger et al.~\cite{henzinger2024unifying} consists of lower-triangular real-valued matrices. The matrix 
$\widehat L = \begin{pmatrix} \widetilde L_r & \widetilde L_c \end{pmatrix} \in \real^{n \times (4n)}$ 
is still a rectangular, real matrix, but we need a factorization where the left matrix is a lower-triangular matrix. This can be done by using the Gram-Schmidt orthogonalization trick of Denisov et al.~\cite{mcmahan2022private} on $\widehat L^\top$ to get two matrices $L^\top$ and $Q^\top$ such that $L^\top$ is upper-triangular and $Q^\top$ is an orthogonal matrix. It follows that $L$ is lower-triangular and that $LQ = \widehat L$. We then set $R = Q^\top \widehat R$. Again from the fact that $Q$ is an orthogonal matrix, we have $\gamma_{(p)}(\decayingmatrix) \leq \ptrace{L}\colnorm{R} = \mathsf{Tr}_{\op p}(\widehat L)\colnorm{\widehat R}$. The details of this approach are given in \Cref{alg:factorization}.

The proof that factors, $L$ and $R$, used in \Cref{alg:factorization} also satisfy the same upper bound on $\gamma_{\op{(p)}}(\decayingmatrix)$  follows from basic properties of generalized $p$-trace and absolute values of complex numbers. More formally, we show the following to complete the proof on the upper bound on $\gamma_{\op{(p)}}(\decayingmatrix)$ as required in Theorem \ref{thm:mainupperboundgamma}.
 
\begin{lem}
\label{lem:gammanormequivalence}
Let $\widetilde L,\widetilde R,L$, and $R$ be as defined in \Cref{alg:factorization}. Then for $p\in [2,\infty)$
$$\ptrace{L} \colnorm{R} =\mathsf{Tr}_{\op p}(\widetilde L) \colnorm{\widetilde R}.$$ 
\end{lem}
\begin{proof}
    Let $\widetilde L = \widetilde L_r + \iota \widetilde L_c$ be the complex decomposition of the matrix $\widetilde L$ and $\widetilde R = \widetilde R_r + \iota \widetilde R_c$ be the complex decomposition of the matrix $\widetilde R$ such that $\widetilde L_r,\widetilde L_c, \widetilde R_c, \widetilde R_r \in \real^{n \times 2\streamlength}$. Then we can define the matrices $\widehat L,\widehat R$ as in \cref{eq:widehatFactors}.    Since $|a+\iota b| = \sqrt{a^2+b^2}$, we have the following:
    \begin{align}
    \label{eq:widetildeLtowidehatL}
    \ptrace{\widehat L} = \ptrace{\widetilde L} \quad \text{and} \quad \colnorm{\widehat R} = \colnorm{\widetilde R}.    
    \end{align}

The Gram-Schmidt decomposition of $
\widehat{L}^\top$  returns an orthogonal matrix $Q^\top$ and an upper triangular matrix $L^\top$ with $(\widehat{L})^\top = Q^\top L^\top$. Thus,  $L$ is lower triangular,
$\widehat L= L Q$ and $\ptrace{L} = \ptrace{LQ} =
%\colnorm{L^\top} = \colnorm{Q^\top L^\top} = \colnorm{\widehat L^\top} = 
\ptrace{\widehat L}$, as $QQ^\top=\mathbb I$. 
We set $R = Q^\top \widehat R$  and, as $\colnorm{\cdot}$ is the maximum $\ell_2$ norm of all the columns,  we have 
\begin{align}
\label{eq:widehatLtoL}
\begin{split}
    \ptrace{L}  = \ptrace{\widehat L} \quad \text{and} \quad
    \colnorm{R} = \colnorm{Q^\top \widehat R} = \colnorm{\widehat R}.
\end{split}
\end{align}
Combining \cref{eq:widetildeLtowidehatL} and \cref{eq:widehatLtoL} completes the proof of \Cref{lem:gammanormequivalence}.
\end{proof}

\section{Bounds for Special Matrices}
\label{sec:specialMatrices}
In this section we show how to achieve improved upper bounds on the additive error $\ellperror$ for any $p \in [2,\infty)$ for continual release algorithms based on the matrix mechanism for concrete matrices $M_f$.
Specifically, we prove Corollary~\ref{cor:continualcounting}.

We start with the continual counting matrix. 
In this case, $f(t)=1$ for all $0 \leq t\leq \streamlength-1$. Therefore, $m_f(x)= 1 + x+ x^2 + \cdots + x^{\streamlength-1}$. First note that for integer $\ell$
\begin{align}
    \omega^{\ell n} = e^{\iota \ell \pi} = \cos(\ell \pi) + \iota \sin (\ell \pi) = \begin{cases}
    1 & \text{ if } \ell \mod 2 = 0 \\
    -1 &  \text{ if } \ell \mod 2  = 1
\end{cases}.
\label{eq:powerofomega}
\end{align}
i.e., $1-\omega^{\ell n} = 0$ for even $\ell$ and 2 otherwise.
Using \Cref{lem:1-omega} shows that 
\begin{align*}
    \left| 1 - \omega^\ell \right|  = 2 \left|\sin\paren{{\pi \ell \over 2\streamlength}}  e^{-\iota \left({\pi\over 2} - {\pi \ell \over 2\streamlength} \right)} \right| 2 \left|\sin\paren{{\pi \ell \over 2\streamlength}} \right|.
\end{align*}
Plugging in the bound, we get from Theorem~\ref{thm:mainupperbound}, 
\begin{align*}
    \gamma_{\op 2}(M_\counting) &\leq {1 \over 2 \streamlength} \sum_{\ell=0}^{2\streamlength-1} \left|\sum_{k=0}^{\streamlength-1}  \omega^{k\ell} \right| = {1 \over 2} + {1\over 2\streamlength} \sum_{\ell=1}^{2\streamlength-1} \left|{1 - \omega^{\ell n}\over 1 - \omega^\ell} \right| = 
    {1 \over 2} + {1 \over n} \sum_{\ell \text{ is odd}}^{2\streamlength-1}  \left|{1\over 1 - \omega^\ell} \right| \\
    &= {1 \over 2} + {1 \over 2\streamlength} \sum_{\ell \text{ is odd}}^{2\streamlength-1}  \left| {1 \over \sin\paren{{\pi \ell \over 2\streamlength}}} \right| = {1 \over 2} + \underbrace{{1 \over 2\streamlength} \sum_{\ell =1}^{n} \left| {1 \over \sin\left({\pi (2\ell-1)\over 2\streamlength} \right)} \right|}_{\gamma_\streamlength}
\end{align*} 
This is the same bound achieved by non-constructive methods in Mathias~\cite{mathias1993hadamard} (see equation (2.8) in Dvijotham, McMahan, Pillutla, Steinke, and Thakurta~\cite{dvijotham2024efficient}). 
Mathias~\cite{mathias1993hadamard} showed that the term $\gamma_\streamlength$ is an increasing sequence { in $n \in \mathbb N$}. Therefore, $\min_{n \in \mathbb N}\{ \gamma_\streamlength\}=\gamma_1$ and $\max_{n \in \mathbb N}\{ \gamma_\streamlength \} = \gamma_\infty$. Now $$\gamma_1 =  {1 \over 2} \sum_{\ell \text{ is odd}}^{1}  \left| {1 \over \sin\paren{{\pi \ell \over 2}}} \right| = {1 \over 2} \left| {1 \over \sin\paren{{\pi \over 2}}} \right| = {1\over 2}.$$ 
 
Using \Cref{lem:sumCosecantfunctionlimit}, we now have that
\[
\gamma_{\op 2}(M_\counting) \leq {1 \over 2} + {1 \over 2\streamlength} \sum_{\ell =1}^{n} \left| {1 \over \sin\left({\pi (2\ell-1)\over 2\streamlength} \right)} \right| \leq 1 +  {1 \over \pi} \ln \paren{n}.
\]

This is slightly smaller than the upper bound
of 
 $ 1 + {\ln(\streamlength) \over \pi}$ for this expression given in Dvijotham, McMahan, Pillutla, Steinke, and Thakurta~\cite{dvijotham2024efficient} as $\ln(2\streamlength/\pi)  = \ln(\streamlength) - \ln(\pi/2)< \ln(\streamlength).$

Finally, using the same calculation, Theorem \ref{thm:mainupperboundgamma} gives for any $p \in [2,\infty)$
\begin{align*}
    \gamma_{\op{(p)}}(M_\counting) = n^{1/p} \cdot \gamma_{\op 2}(M_\counting)
    \leq n^{1/p}\paren{\countingupperbound},
\end{align*}   
which is an improvement over Liu et al.  \cite{liu2024optimality}.

The bounds for the sliding window model follow similarly as above after noting that the polynomial for the sliding window model is $m_f(x) = 1+x+x^2+\cdots+x^{W-1}$.
Specifically, using the same calculation and replacing $n$ by $W$, 
Theorem \ref{thm:mainupperboundgamma} gives for any $p \in [2,\infty)$, 
\[
\gamma_{\op{(p)}}(M_\sliding)  \leq n^{1/p}\paren{\slidingupperbound},
\]
which is the first bound on $\gamma_{\op{(p)}}(M_\sliding)$ for any $p \in [2,\infty)$.

{
Finally, for the $b$-\mname matrix, with $b \in \mathbb N$, let $k \in \real$  be such that $n=kb$. 
As in the case of Fichtenberger et al. \cite{fichtenberger2023constant}, we assume below that $k$ is an integer; if that is not the case, round up $n$ to the closest multiple of $b$, which implies replacing $n$ by at most $n+b-1$ in the upper bound on $\gamma_{\op 2}(M_\multi)$.
%each example participates at most $k$ times with each participation exactly b steps apart, setting, 
We have $M_\multi = M'_\counting \otimes \I_{b \times b}$, where $M'_\counting \in \set{0,1}^{\streamlength/b \times \streamlength/b}$ is the prefix sum matrix defined for a stream of length $n/b$. Using the fact that $\gamma_{\op 2}(A \otimes B) = \gamma_{\op 2}(A) \cdot \gamma_{\op 2}(B)$~\cite{haagerup1980decomposition} and $\gamma_{\op 2}(U)=1$ for any unitary matrix $U$, we have 
\[
\gamma_{\op 2}(M_\multi) = \gamma_{\op 2} (M'_\counting) \gamma_{\op 2}(\mathbb I_{b \times b}) = \gamma_{\op 2} (M'_\counting) \leq  1 + {1\over \pi} \log\paren{\streamlength \over b}.
\]
}

{
\section{Tightness of \texorpdfstring{Theorem~\ref{thm:mainupperboundgamma}}{Lg}}
\label{sec:optimality}
\label{sec:optimalityCounting}
\label{sec:slidingOptimality}
In this section, we show that Theorem~\ref{thm:mainupperboundgamma} leads to an almost tight bound for the $M_\counting$ {and $M_\multi$} matrix. We start with $\gamma_{\op{(p)}}(M_\counting)$. 
Using Liu, Upadhyay, and Zou~\cite[Lemma 16]{liu2024optimality}, we have 
\[
\gamma_{\op{(p)}}(M_\counting) \geq {1 \over n^{1-1/p}} \norm{M_\counting}_1,
\]
where $\norm{\cdot}_1$ is the Schatten-$1$ norm. It was shown in Elliot~\cite{elliott1953characteristic} that, for $M_\counting$, we have 
\[
\norm{M_\counting}_1 = \frac{1}{2} \sum_{i=1}^\streamlength \left\vert\csc \paren{\frac{(2i-1)\pi}{4\streamlength+2}} \right\vert.
\]

From this, using the calculation in Henzinger, Upadhyay, and Upadhyay~\cite[Section 5.2]{henzinger2023almost} and Liu et al.~\cite{liu2024optimality}, we have 
\[
\gamma_{\op{(p)}} (M_\counting)\geq n^{1/p}  \paren{ {2\over \pi} + {1\over \pi}\ln\paren{2\streamlength+1 \over 5} + {{\ln(2\streamlength+1) \over 2\streamlength \pi}}}.
\]

In contrast, Theorem \ref{thm:mainupperboundgamma} gives 
\[
\gamma_{\op{(p)}} \leq n^{1/p}\paren{\countingupperbound}.
\]

We next show a simple lower bound for the $\gamma_2$ norm of $M_\sliding$.
\begin{theorem}3e4
\label{thm:lowerSlidingWindow}
    Let $M_\sliding$ be the matrix as defined in \cref{eq:sliding}. Then 
    \begin{align}
    \gamma_{\op{2}}(M_\sliding) \geq \max \set{\paren{W+1 \over 2W} \sum_{i=1}^W \left| \csc \paren{(2i-1)\pi \over 2W} \right|, {\ln((2W+1)/3) + 2 \over \pi}}.
    \end{align}
\end{theorem}
\begin{proof}
    Consider the $W \times W$ submatrix of $M_\sliding$. Let it be $M_\sliding[:W,:W]$. Using Fact \ref{claim:monotonicity}, we have 
    $\gamma_{\op 2}(M_\sliding) \geq \gamma_{\op 2}(M_\sliding[:W,:W]).$  
    Now note that $M_\sliding[:W,:W]$ is equivalent to a $W \times W$ matrix of the form  $M_\counting$. Setting $n=W$ in Theorem \ref{thm:mathias1993hadamard} and \cref{eq:matousek2020} gives us the bound. 
\end{proof}

Using the monotonicity of $\gamma_{\op{(p)}}$-norm under the addition of rows and columns, we also have 
\[
\gamma_{\op{(p)}}(M_\sliding)  \geq n^{1/p}  \paren{ {2\over \pi} + {1\over \pi}\ln\paren{2W+1 \over 5} + {{\ln(2W+1) \over 2\streamlength \pi}}}.
\]

In contrast, Theorem \ref{thm:mainupperboundgamma} gives 
\[
\gamma_{\op{(p)}}(M_\sliding)  \leq n^{1/p}\paren{\slidingupperbound}.
\]

\section{Conclusion and Open Problems}
We presented a new technique for factorizing matrices and bounding the $\gamma_p$ norms of Toeplitz matrices using group algebra. We used this technique to give improved bounds for the additive error for differentially private algorithms for various weighted prefix sum problems in the continual and the sliding window setting.
However, as our approach is very general, we believe that it can be used to achieve further novel algorithms, for example, in the regime where space as well as running time is bounded, which is the setting that due to its use in differentially private learning currently receives a lot of interest (see e.g.~\cite{dvijotham2024efficient}).

Our approach, however, requires the knowledge of the streamlength $\streamlength$ in contrast to previous explicit factorizations~\cite{henzinger2023almost,henzinger2024unifying}. In practical scenarios, we ideally would like to have a factorization for an unknown length stream. A natural question is whether we can compute explicit factorization where we do not know the stream length and match (or improve) the bounds achieved in this paper.

\section{Acknowledgements}
\erclogowrapped{5\baselineskip}Monika Henzinger:  This project has received funding from the European Research Council (ERC) under the European Union's Horizon 2020 research and innovation programme (Grant agreement No.\ 101019564) and the Austrian Science Fund (FWF) grant DOI 10.55776/Z422, grant DOI 10.55776/I5982, and grant DOI 10.55776/P33775 with additional funding from the netidee SCIENCE Stiftung, 2020–2024.

Jalaj Upadhyay's research was funded by the Rutgers Decanal Grant no. 302918 and an unrestricted gift from Google. This work was done in part while visiting the Institute of Science and Technology Austria (ISTA).

The authors would like to thank Sarvagya Upadhyay for the initial discussion and feedback on the early draft of the paper. The authors would like to thank the anonymous reviewers, Brendan McMahan, Nikita Kalinin, Jingcheng Liu, and Abhradeep Thakurta for the discussions that helped improve the presentation of the final version of the paper.

\bibliographystyle{plain}
\bibliography{privacy}

\appendix

\section{Auxiliary Lemmata}
\begin{fact}
    [Half angle formula]
    \label{lem:halfAngle}
    Let $\theta \in \real$. Then 
    $\cos(2\theta) =  1-2\sin^2(\theta) $ and $
        \sin(2\theta) = 2\sin(\theta) \cos(\theta).$
\end{fact}

\begin{lem}\label{lem:1-omega}
    Let $\omega$ be the generator of the set of unit roots of order $2\streamlength$ and let $\ell$ be an integer. Then
    $$1 - \omega^\ell = 2\sin\paren{{\pi \ell \over 2\streamlength}}  e^{-\iota \left({\pi\over 2} - {\pi \ell \over 2\streamlength} \right)}.$$
\end{lem}
\begin{proof}
Recall that $\sin^2(\ell \pi/(2\streamlength)) + \cos^2 (\ell \pi/(2\streamlength))=1$, and Fact \ref{lem:halfAngle} shows that
$\cos(\ell \pi/n) = \cos^2(\ell\pi/(2\streamlength)) - \sin^2(\ell \pi/(2\streamlength))$, and 
$\sin(\ell \pi/n) = 2\cos(\ell\pi/(2\streamlength)) \sin(\ell\pi/(2\streamlength))$. Next 
\begin{align*}
\sin \paren{{\pi \ell \over 2\streamlength}} &= \cos\left({\pi\over 2} - {\pi \ell \over 2\streamlength} \right) =\cos\left(-\left({\pi\over 2} - {\pi \ell \over 2\streamlength} \right) \right)  \quad \text{and} \\\ 
\cos\paren{{\pi \ell \over 2\streamlength}} &= \sin\left({\pi\over 2} - {\pi \ell \over 2\streamlength} \right) =-\sin\left(-\left({\pi\over 2} - {\pi \ell \over 2\streamlength} \right) \right).    
\end{align*}

Therefore, we can write 
\[
\sin \paren{{\pi \ell \over 2\streamlength}} - \iota \cos\paren{{\pi \ell \over 2\streamlength}} = \cos\left(-\left({\pi\over 2} - {\pi \ell \over 2\streamlength} \right) \right) + \iota \sin\left(-\left({\pi\over 2} - {\pi \ell \over 2\streamlength} \right) \right)= e^{-\iota \left({\pi\over 2} - {\pi \ell \over 2\streamlength} \right)}.
\]

Recall that $\omega = \cos(\pi/n) + \iota \sin(\pi/n)$. \Cref{lem:1-omega} follows from following calculation:
    \begin{align*}
    1 - \omega^\ell &= 1 - \cos(\pi \ell/n) - \iota \sin(\pi \ell/n) \\
    & =
     \sin^2\paren{{\pi \ell \over 2\streamlength}} + \cos\paren{{\pi \ell \over \streamlength}} + \sin^2\paren{{\pi \ell \over 2\streamlength}}  - \cos(\pi \ell/n) - \iota \sin(\pi \ell/n)
    \\
    &=2 \sin^2\paren{{\pi \ell \over 2\streamlength}} - 2 \iota \sin\paren{{\pi \ell \over 2\streamlength}} \cos\paren{{\pi \ell \over 2\streamlength}} 
    = 2\sin\paren{{\pi \ell \over 2\streamlength}}  e^{-\iota \left({\pi\over 2} - {\pi \ell \over 2\streamlength} \right)}.
\end{align*}
\end{proof}

Mathias~\cite{mathias1993hadamard} claimed the following. We show it for the sake of completion. 
\begin{lem}
\label{lem:sumCosecantfunctionlimit}
    We have the following
    \[
    {1 \over 2\streamlength} \sum_{\ell= 1}^{\streamlength}  \left| {\csc\paren{{(2\ell-1)\pi  \over 2\streamlength}}} \right|  \leq {1 \over 2} +  {1 \over \pi } \ln \paren{n}. %+ O\paren{1 \over n^2} .
    \]
\end{lem}
\begin{proof}
Since for any $\theta$,  $1 + |\cot(\theta)| \ge 1 + {\cos^2(\theta) \over |\sin (\theta)|  } = 
{|\sin(\theta)| + \cos^2(\theta) \over |\sin  (\theta)|} \geq {1 \over |\sin (\theta)|} = |\csc(\theta)|$, we have
\begin{align*}
{1 \over 2\streamlength}\sum_{\ell= 1}^{\streamlength}  \left| \csc\paren{{(2\ell-1)\pi  \over 2\streamlength}} \right| 
    &\leq {1 \over 2} + {1 \over 2\streamlength}\sum_{\ell= 1}^{\streamlength}  \left|  \cot\paren{{(2\ell-1)\pi  \over 2\streamlength}} \right| 
\end{align*}

To get the bound, we have to bound 
\[
\gamma(n):= {1 \over 2\streamlength} \sum_{\ell= 1}^{\streamlength}  \left| {\csc\paren{{(2\ell-1)\pi  \over 2\streamlength}}} \right| .
\]
 Mathias~\cite{mathias1993hadamard} has shown that $\gamma(n)$ is an increasing function. Therefore, an upper bound follows if we bound $\lim_{n\to \infty} {\gamma(n) \over \ln(n)}$ since $\gamma(1)={1/2}$. As $\streamlength \to \infty$, summation can be replaced by the Reimman integral. Therefore, 
\begin{align}
\label{eq:limits}
     \lim_{\streamlength \to \infty}{1 \over 2\streamlength {\ln(n)}}\sum_{\ell= 1}^{\streamlength}  \left|  \cot\paren{{(2\ell-1)\pi  \over 2\streamlength}} \right| =  {1 \over 2\streamlength{\ln(n)}}\int\limits_{1}^{\streamlength}  \left|  \cot\paren{{(2x-1)\pi  \over 2\streamlength}} \right| \mathsf{d}x.
\end{align}

Using the fact that $\mathsf d(\sin (\theta)) = \cos(\theta) \mathsf{d}\theta$, we have   
\begin{align*}
{1 \over 2\streamlength\ln(n)}\int\limits_{1}^{\streamlength}  \left|  \cot\paren{{(2x-1)\pi  \over 2\streamlength}} \right| \mathsf{d}x
    &= {1 \over 2\pi \ln(n) } \sparen{\op{sign}\paren{\cot\paren{(2x-1)\pi \over 2\streamlength}} \ln \left|\sin \paren{(2x-1)\pi \over 2\streamlength} \right| }_{1}^n   
\end{align*}

When $x=n$, $$\op{sign}\paren{\cot\paren{(2x-1)\pi \over 2\streamlength}} = \op{sign}\paren{\cot\paren{\pi  - {\pi \over 2\streamlength}}} = -1.$$ 

On the other hand, when $x=1$, 
$$\op{sign}\paren{\cot\paren{\pi \over 2\streamlength}} = 1.$$ 

Further, $\sin \paren{(2\streamlength-1)\pi \over 2\streamlength} = \sin \paren{\pi  - {\pi \over 2\streamlength}} =   \sin \paren{\pi \over 2\streamlength}$. 
Thus,  
$$\left|\sin \paren{(2\streamlength-1)\pi \over 2\streamlength } \right| = \left|\pm \sin \paren{\pi \over 2\streamlength} \right| = \sin \paren{\pi \over 2\streamlength}$$ and, 
as $0 \le {\pi  \over 2\streamlength} \leq {\pi \over 2}$. Thus,
\begin{align*}
{1 \over 2\streamlength \ln(n)}\int\limits_{1}^{\streamlength}  \left|  \cot\paren{{(2x-1)\pi  \over 2\streamlength}} \right| \mathsf{d}x
    & = {1 \over 2\pi \ln(n) } \paren{- \ln \sin \paren{\pi \over 2\streamlength} - \ln \sin \paren{\pi \over 2\streamlength}  }      
    =  {1 \over \pi \ln(n)} \ln \paren{\csc\paren{\pi \over 2\streamlength} }.  
\end{align*}
  
The Taylor series expansion of $\ln(\csc(x))$ around $x=0$ is  
\[
\ln(\csc(x)) = \ln(1/x) + {x^2 \over 6} + {x^4 \over 180} + {x^6 \over 2835} + O(x^8). 
\]
This implies that 
\begin{align*}
{1 \over 2\streamlength \ln(n)}\int\limits_{1}^{\streamlength}  \left|  \cot\paren{{(2x-1)\pi  \over 2\streamlength}} \right| \mathsf{d}x &= {1\over \pi \ln(n)} \paren{ \ln(2\streamlength/\pi) + {\pi^2 \over 24n^2} +  O\paren{1\over n^4} } 
\\ & = {1\over \pi \ln(n)} \paren{\ln(n) + \ln(2/\pi) + {\pi^2 \over 24n^2} +  O\paren{1\over n^4} }
\end{align*}
Therefore, as $n \to \infty$, we have   
\begin{align*}
\lim_{n\to \infty} {1 \over 2\streamlength{\ln(n)} }\int\limits_{1}^{\streamlength}  \left|  \cot\paren{{(2x-1)\pi  \over 2\streamlength}} \right| \mathsf{d}x
    &\to {1 \over \pi } 
\end{align*}
completing the proof.
\end{proof}

\end{document}